\documentclass[a4paper,11pt]{amsart}
\usepackage{amsfonts,amssymb,epsfig,latexsym}
\usepackage{amsmath}
\usepackage[latin1]{inputenc}
\usepackage{color,layout}
\usepackage{ifthen}

\usepackage{graphicx}
\usepackage{color}

\newtheorem{theorem}{Theorem}[section]
\newtheorem{lemma}[theorem]{Lemma}
\newtheorem{proposition}[theorem]{Proposition}

\newtheorem{remark}[theorem]{Remark}

\def\square{\hbox{\vrule\vbox{\hrule\phantom{o}\hrule}\vrule}}

\newcommand{\be}{\begin{equation}}
\newcommand{\ee}{\end{equation}}

\numberwithin{equation}{section}

\newcommand{\Z}{\mathbb{Z}}
\newcommand{\R}{\mathbb{R}}
\newcommand{\C}{\mathbb{C}}
\newcommand{\W}{{\mathcal W}}

\newcommand{\cO}{{\mathcal O}}

\newcommand{\lan}{\langle}
\newcommand{\ran}{\rangle}
\newcommand{\pal}{\parallel}

\newcommand{\re}{{\rm Re}\hskip 1pt }
\newcommand{\im}{{\rm Im}\hskip 1pt }
\newcommand{\ord}{{\mathcal O}}
\newcommand{\ai}{{\rm Ai}\,}

\def\eq#1{(\ref{#1})}

\newtheorem{thm}{Theorem}[section]

\newtheorem{lem}{Lemma}[section]
\newtheorem{rem}{Remark}[section]

\numberwithin{equation}{section}

\begin{document}

\title{Widths of resonances above an energy-level crossing}
\author{S.~Fujii\'e${}^1$, A.~Martinez${}^2$ and T.~Watanabe${}^3$}



\maketitle
\addtocounter{footnote}{1}
\footnotetext{{\tt\small  Department of Mathematical Sciences, Ritsumeikan University, 1-1-1 Noji-Higashi, Kusatsu, 525-8577,  Japan, 
fujiie@fc.ritsumei.ac.jp} }  
\addtocounter{footnote}{1}
\footnotetext{{\tt\small Universit\`a di Bologna,  
Dipartimento di Matematica, Piazza di Porta San Donato, 40127
Bologna, Italy, 
andre.martinez@unibo.it }}  
\footnotetext{{\tt\small  Department of Mathematical Sciences, Ritsumeikan University,  1-1-1 Noji-Higashi, Kusatsu, 525-8577,  Japan,
t-watana@se.ritsumei.ac.jp} }  

\begin{abstract}
We study the existence and location of the resonances of a $2\times 2$ semiclassical system of coupled Schr\"odinger operators, in the case where the two electronic levels cross at some point, and one of them is bonding, while the other one is anti-bonding. Considering energy levels just above that of the crossing, we find the asymptotics of both the real parts and the imaginary parts of the resonances close to such energies. This is a continuation of our previous works where we considered energy levels around that of the crossing.
\end{abstract}
\pagebreak

\section{Introduction}
In this paper we continue our study of the resonances for a $2\times 2$ semiclassical system of coupled one dimensional Schr\"odinger operators. Let us recall (see, e.g. \cite{KMSW}) that such systems appear in a natural way after a Born-Oppenheimer reduction of diatomic molecular Hamiltonians (in which case, the semiclassical parameter $h$ actually represents the square root of the quotient between the electronic and nuclear masses).

The situation we investigate is that of two electronic levels that cross at some real point, corresponding to an energy (say, 0) that is inside the continuous spectrum of the Hamiltonian (see Figure 1). 
Such a phenomenon occurs when some energy levels of the molecule cross each other, and this may happen even for diatomic molecules (see, e.g., \cite{Be, DiVi, Le, LeSu} and references therein), in which case, due to the rotational invariance of the system, the dimension of the space variable $x$ can be  reduced to 1 (so that $x$ actually represents the distance between the two atoms): see \cite{KMSW, MaSo}. Then, the eigenprojectors can generically be assumed to be smooth, and the resulting matrix-potential diagonalizes smoothly, leading to the model we are studying.

Let us also observe that the case of avoided-crossing with a gap of order smaller than $h$ also enters our model (in this case, the gap of the avoided-crossing may be included in the coefficient $r_0$ of the interaction $W$: see Assumption (A5)).

In \cite{FMW1}, we have considered energies very close to 0 (mainly, of size $\ord (h^{\frac23}$)), and we have proved the existence of resonances there, and given a sharp estimate on their widths in the case of an elliptic interaction at the crossing point.

Then, this result has been extended in \cite{FMW2} to the case of non elliptic interactions, that is, more precisely, to the physical case where the interaction consists of a (non zero) vector field.

In both cases, the main technique consists in constructing global WKB solutions for the system. This is made possible thanks to appropriate estimates on the fundamental solutions of the scalar Schr\"odinger operators involved in the problem. (Let us recall that the usual WKB constructions made for scalar operators cannot be extended to systems in general.)

Here, we consider the same situation, but this time we investigare the resonances $E=E(h)$ that have a real part close to some $E_0>0$ and an imaginary part $\ord (h)$. Of course, we still use the WKB solutions constructed in \cite{FMW1}, but since the region where they oscillate  overlap, the behaviours of the fundamental solutions are not as good as in the cases of \cite{FMW1, FMW2}. However, they are sufficient to prove the existence of resonances with a rough estimate of their location.

In order to compute in a precise way their asymptotics, we need to implement microlocal techniques in our methods. Namely, we need to specify the microlocal behaviour of the resonant states near the two crossing points (in phase-space) of the two characteristic surfaces of the problem (see Figure 2). To do so, around each of these points we construct two microlocal bases of solutions with specific properties concerning their microsupports, and we express any resonant state in these two basis. After that, the connection formulas between the bases permit us to obtain a very precise asymptotic expansion of the resonant states on the outgoing branch of the characteristic set, from which sharp estimates on the resonances (both their real parts and their widths) can be derived.

The microlocal constructions are made in a spirit similar to that of \cite{Ab} (that is, mainly by reducing the operator to a normal form), but with a particular attention on the choice of the Fourier integral operators used in the reduction.

Let us observe that our final result gives resonance widths of order $h^2$, that should be compared with the order $h^{\frac53}$ found in \cite{FMW1}. This difference can be explained by the different geometries of the characteristic sets. Indeed, in both cases they are constituted by  two curves, but in the current situation these curves cross transversally, while in \cite{FMW1} they are tangent to each other. In some sense, it becomes natural to think that in \cite{FMW1} the particle escapes more easily than here, which makes its life-time shorter, and  thus its resonance width larger.

Let us also mention the work \cite{As} where a similar problem is considered, but for negative energies (in which case the resonance widths are exponentially small).

The content of the paper is as follows: In section 2, we state our main theorem.
In section 3, we construct outgoing solutions to our system using the method established in \cite{FMW1} and \cite{FMW2}.
In section 4, we compute the wronskian between these outgoing solutions to show the existence and a rough estimate of the location of resonances. In section 5, we study the structure of the space of microlocal solutions near  crossing points.
In section 6, we apply this result to compute the asymptotic behaviour of the resonant state globally in the phase space.
Finally in section 7, we complete the proof of our main theorem.

{\bf Acknowledgements} This research was  supported by  the
University of Bologna during the visit of the first author in 2018 as well as the JSPS grant-in-aid for scientific research. We also thank  Y.Tsutsumi for his hospitality in Kyoto University where part of this work was done and M. Assal for his valuable remarks after reading the text carefully.

\begin{center}
\scalebox{0.4}{
\includegraphics{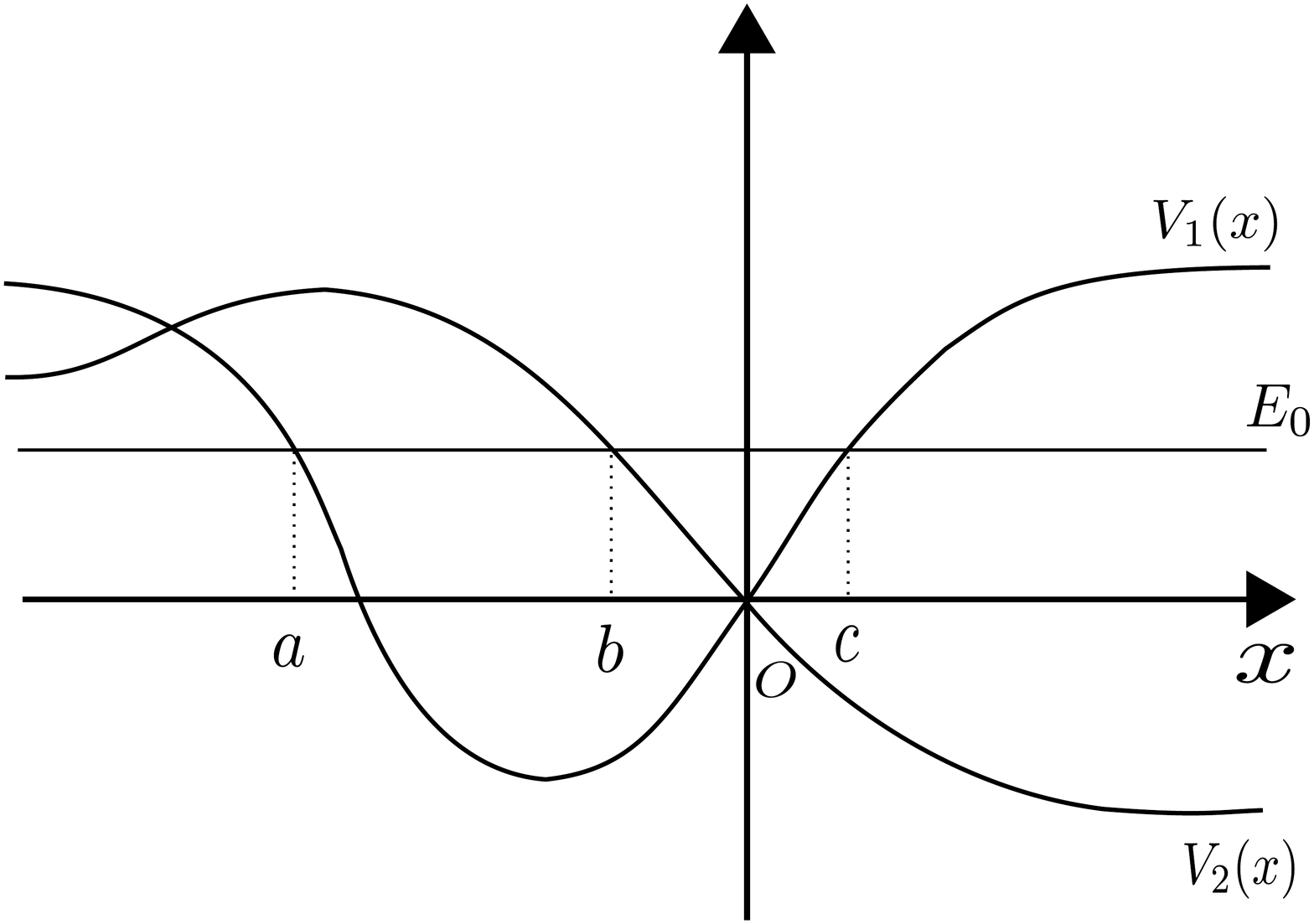}
}
\centerline{Figure 1: The two potentials}
\end{center}

\begin{center}
\scalebox{0.4}{
\includegraphics{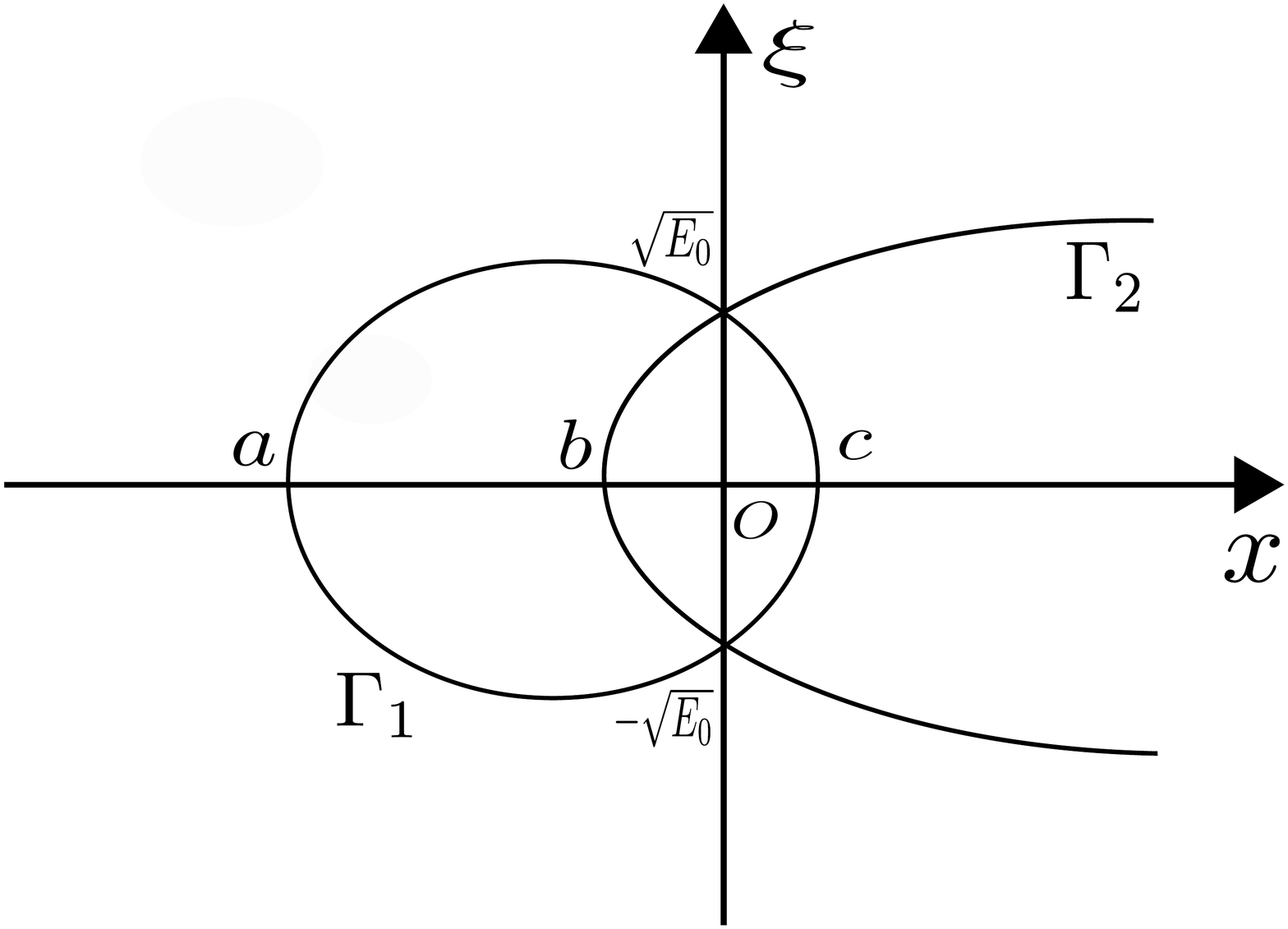}
}
\centerline{Figure 2: Phase-space picture}
\end{center}

\vskip 0.3cm

\section{Assumptions and results}

We consider a $2\times 2$ Schr\"odinger operator of the type,
\begin{equation}
\label{sch}
\begin{aligned}
Pu &= Eu,\qquad
P &= \left(
\begin{matrix}
 P_1 & hW\\
hW^* & P_2
\end{matrix}
\right),
\end{aligned}
\end{equation}
where $D_x$ stands for $-i\frac{d}{d x}$, $P_j=h^2D_x^2 + V_j(x)$ ($j=1,2$), $W=W(x,hD_x)$ is a semiclassical differential operator, and $W^*$ is the formal adjoint of $W$.
We study the asymptotic distribution of resonances in the semiclassical limit $h\to 0_+$ in a neighborhood of a fixed real energy $E_0$.

We suppose the following conditions on the potentials $V_1(x), 
V_2(x)$ (see Figure 1),  the positive number $E_0$  and  the interaction
$W(x,hD_x)$:

{\bf Assumption (A1)}
$V_1(x)$, $V_2(x)$ 
 are real-valued analytic functions on $\R$, and extend to holomorphic functions
 in the complex domain,
$$
{\mathcal S}=\{x\in\C\,;\,|\im\, x|<\delta_0\lan\re \,x\ran\}
$$
where $\delta_0>0$ is a constant, and $\lan t\ran:=(1+|t|^2)^{1/2}$.

{\bf Assumption (A2)} For $j=1,2$, $V_j$ admits limits as $\re\, x\to \pm\infty$ in ${\mathcal S}$, and they satisfy,
$$
\begin{aligned}
\lim_{{\re\,x\to -\infty}\atop{x\in {\mathcal S}}}  V_1(x)>E_0\, ;\, \lim_{{\re\,x\to -\infty}\atop{x\in {\mathcal S}}} V_2(x)>E_0\, ;\\
\lim_{{\re\,x\to +\infty}\atop{x\in {\mathcal S}}} V_1(x)>E_0\, ;\, \lim_{{\re\,x\to +\infty}\atop{x\in {\mathcal S}}} V_2(x)<E_0.
\end{aligned}
$$

{\bf Assumption (A3)} There exist three  numbers $a<b<0<c$ such that
$$
V_1(a)=V_1(c)=V_2(b)=E_0,
$$
$$
V_1'(a) < 0,\qquad V_1'(c) > 0,\qquad V_2'(b) < 0.
$$
and that
$$
\begin{array}{ll}
V_1>E_0\text{ on }(-\infty, a)\cup (c,+\infty),\quad
&V_1<E_0\text{ on }(a,c), \\[8pt]
V_2>E_0\text{ on }(-\infty, b),\quad
&V_2<E_0\text{ on }(b,+\infty).
\end{array}
$$

{\bf Assumption (A4)}
The set $\{x\in \R; V_1(x)=V_2(x)\, , \, V_1(x)\leq E_0\, ,\, V_2(x)\leq E_0\}$ is reduced to $\{0\}$, and one has $V_1(0)=V_2(0)=0$, $V'_1(0)>0$, $V_2'(0)<0$.

In particular, in the phase-space, the characteristic sets 
$\Gamma_j:=\{\xi^2 +V_j(x)=E_0\}$ ($j=1,2$) intersect transversally at 
$(0,\pm\sqrt{E_0})$ (see Figure 2).

{\bf Assumption (A5)}
$W(x,hD_x)$ is a first order differential operator,
$$
W(x,hD_x)=r_0(x)+ir_1(x)hD_x,
$$
where $r_0$ and $r_1$ are bounded analytic function on ${\mathcal S}$,
are real on the real, and such that $W$ is elliptic at the 
crossing points $(0,\pm\sqrt{E_0})$, that is,
$$
(r_0(0),r_1(0)) \not=(0,0).
$$

\vskip 0.3cm
In this situation, in a neighbourhood of the energy $E_0$, the spectrum of $P$ is essential only, and the resonances of $P$ can be defined, e.g., as the values $E\in\C$ such that the equation $Pu=Eu$ has a non trivial outgoing solution $u$, that is, a non identically vanishing solution such that, for some $\theta >0$ sufficiently small, the function $x\mapsto u(xe^{i\theta})$ is in $L^2(\R)\oplus L^2(\R)$ (see, e.g., \cite{AgCo, ReSi}). Equivalently, the resonances are the eigenvalues of the operator $P$ acting on $L^2(\R_\theta)\oplus L^2(\R_\theta)$, where $\R_\theta$ is a complex distortion of $\R$ that coincides with $e^{i\theta}\R$ for $x\gg 1$ (see, e.g., \cite{HeMa}). We denote by ${\rm Res}(P)$ the set of these resonances.

For $E\in \C$ close enough to $E_0$, we define the action,
$$
{\mathcal A}(E):= \int_{a(E)}^{c(E)}\sqrt{ E-V_1(t)} \, dt,
$$
where $a(E)$ (respectively $c(E)$) is the unique solution of $V_1(x)=E$ close to $a$ (respectively close to $c$). 
In this situation, ${\mathcal A}(E)$ is an analytic function of $E$ near $E_0$
 and ${\mathcal A}'(E)$ is strictly positive for any real $E$ near $E_0$.

We also fix $\delta_0>0$ sufficiently small and $C_0>0$ arbitrarily large, and we plan to study the resonances of $P$ lying in the set ${\mathcal D}_h(\delta_0, C_0)$ given by,
\be
{\mathcal D}_h(\delta_0,C_0):= [E_0-\delta_0, E_0+\delta_0]-i[0,C_0h].
\ee

For $h>0$ and $k\in\Z$ such that $(k+\frac12)\pi h$ belongs to ${\mathcal A}( [E_0-2\delta_0, E_0+2\delta_0])$, we set,
\be
\label{defekh}
e_k(h):={\mathcal A}^{-1}\left( (k+\frac12)\pi h\right).
\ee
Then, our main result is,
\begin{thm}\sl
\label{mainth}
Under Assumptions (A1)-(A5), there exists $\delta_0>0$ such that for any $C_0>0$, one has, for $h>0$ small enough
$$
{\rm Res}\,(P)\cap {\mathcal D}_h(\delta_0, C_0)
=\{E_k(h); k\in\Z\}\cap{\mathcal D}_h(\delta_0, C_0)
$$
where the $E_k(h)$'s are complex numbers that satisfy
\begin{align}
\label{reEk}
\re \,E_k(h) &= e_k(h) + {\mathcal O}(h^{2}),\\
\label{imEk}
\im \,E_k(h) &= - C(e_k(h))h^2+ {\mathcal O}(h^{7/3}),
\end{align}
uniformly as $h \to 0$. Here
$$
C(E)=\frac {\pi }{\gamma{\mathcal A}'(E)}
\left |r_0(0)E^{-\frac 14}\sin \left(\frac{{\mathcal B}(E)}{h}  + \frac{\pi}{4}\right) 
+ r_1(0)E^{\frac 14}\cos \left(\frac{{\mathcal B}(E)}{h} + \frac{\pi}{4}\right)\right |^2
$$
with $\gamma := V_1'(0)-V_2'(0) > 0$ and 
$${\mathcal B}(E):=\int_{b(E)}^0\!\!\!\!\sqrt{E-V_2(x)}dx+\int_0^{c(E)}\!\!\!\!\!\sqrt{E-V_1(x)}dx,$$
where $b(E)$ is the unique root of $V_2(x)=E$ close to $b$.
\end{thm}

\begin{rem}
The physical case corresponds to $r_0= 0$ identically (see, e.g. \cite{FMW2}). In that case, the quantity $C(E)$  reduces to,
$$
C(E)=\frac {\pi }{\gamma{\mathcal A}'(E)}
 r_1(0)^2E^{\frac 12}\cos^2 \left(\frac{{\mathcal B}(E)}{h} + \frac{\pi}{4}\right).
$$
\end{rem}

\begin{rem}

Our theorem is valid above the crossing energy $0$ with $h$-independent positive real part.
However, it is interesting to observe the behaviour of the function $C(E)$ for positive $E$ of order $h^{\frac 23}$,
where the asymptotics of resonances are studied in \cite{FMW1} and \cite{FMW2}.
In this energy region, the factors $E^{-\frac 14}$ and $E^{\frac 14}$ concerning $r_0(0)$ and $r_1(0)$ respectively
are of order $h^{-\frac 16}$ and $h^{\frac 16}$, and hence their square together with $h^2$ in \eq{imEk} give
$h^{\frac 53}$ and $h^{\frac 73}$, which coincide with the results in  \cite{FMW1} and \cite{FMW2} respectively.
More precisely, substituting $\lambda_k(h) h^{2/3}$ for $e_k(h)$, we recover the same formula in the above papers
for the width of resonances.

\end{rem}

\begin{rem}
It is important to remark that the function $C(E)$ vanishes on a discrete subset of $\R$ which is determined by the action
${\mathcal B}(E)$.
In particular this set is given by
$$
{\mathcal B}(E)=\left (m-\frac 14\right )\pi h,\quad m\in\Z
$$
if $r_1(0)=0$, and by
$$
{\mathcal B}(E)=\left (m+\frac 14\right )\pi h,\quad m\in\Z
$$
if $r_0(0)=0$. When $e_k(h)$ takes such an value, the first term of the RHS of \eq{imEk} vanishes, which implies that
the width of the corresponding resonance is smaller than $h^2$.
This phenomenon can be seen as an effect of the closed trajectory in the phase space made by $\Gamma_1$ and $\Gamma_2$ (other than $\Gamma_1$ itself).
\end{rem}

\section{Background}
\label{sect3_0}
We fix $\theta >0$ small enough, we use the following notations:
$$
I_L:= (-\infty, 0]\,\, ; \,\, I_R^\theta:=F_\theta ([0, +\infty))\,\, ; \,\,  F_\theta (x):= x+i\theta f(x)
$$
where $f\in C^{\infty}([b, +\infty);\R_+)$, $f(x)=x$ for $x$ large enough, $f(x)=0$ for $x\in[b,x_\infty]$ for some $x_\infty>c$, and $f$ is chosen in such a way that, for any $x\geq x_\infty$, one has,
\be
\label{decItheta}
\im \int_{x_\infty}^{F_\theta (x)}\sqrt{ E-V_2(t)} dt \geq -Ch,
\ee
with some positive constant $C$ (see \cite{FMW1}).

\subsection{Fundamental solutions on $I_L$}
\label{sect3}

For $E\in {\mathcal D}_h(\delta_0,C_0)$ with $\delta_0$ small enough,  let $u_{j,L}^\pm$ be the solutions to 
$(P_j - E)u = 0$ in $(-\infty, 0]$ constructed as in Appendix 2 of \cite{FMW1}. In particular, $u_{j,L}^-$ decays exponentially at $-\infty$, while $u_{j,L}^+$ grows exponentially, and their
Wronskian $\W_{j,L}:=\W[u_{j,L}^-, u_{j,L}^+]$  satisfies 
\begin{equation}
\label{wronsL}
\W_{j,L} = \frac{-2}{\pi h^{\frac23}}(1+\ord(h)) \qquad (h \to 0).
\end{equation}

In the interior of the interval $[\re\, a(E), \re\, c(E)]$, we also have (see \cite[Remark 4.1]{FMW1},
\be
\label{u1L-}
u_{1,L}^-(x) = \frac{2h^{\frac 16}}{\sqrt \pi}(E-V_1(x))^{-\frac14}\cos\left( \frac{{\mathcal A}(E) +\nu_1(x)}{h} -\frac{\pi}4\right) +\ord(h^{\frac76}),
\ee
where $\nu_1(x):= \int_{c(E)}^x \sqrt{E-V_1(t)} dt$.

For any $k\geq 0$ integer, we set,
$$
C^k_b(I_L) := \{ u\, :\, I_L \to \C \mbox{ of class } C^k \,;\, \sum_{0\leq j\leq k}\sup_{x\in I_L}|u^{(j)}(x)| < +\infty \},
$$
equipped with the norm $\pal u \pal_{C^k_b(I_L)} := \sum_{0\leq j\leq k}\sup_{I_L}|u^{(j)}|$, 
and we define a fundamental solution
$$K_{j,L}\, :\, C^0_b(I_L) \to C^2_b(I_L)\qquad (j=1,2),$$
of $P_j - E$ on $I_L$ by setting, for $v\in C^0_b(I_L)$,
\be
\label{eq1}
K_{j,L}[v](x) := \frac{u_{j,L}^+(x)}{h^2 \W_{j,L}} \int_{-\infty}^x  u_{j,L}^-(t)v(t)\,dt + \frac{u_{j,L}^-(x)}{h^2 \W_{j,L}} \int_x^{0}\!\!\!\! u_{j,L}^+(t)v(t)\,dt.
\ee
Then, $K_{j,L}$ satisfies,
$$
(P_j-E)K_{j,L}={\mathbf 1},
$$
and, because of the form of the operator $W$, an integration by parts shows that we also have,
$$K_{j,L}W,\,\, K_{j,L}W^*\, :\, C^0_b(I_L) \to C^0_b(I_L)\qquad (j=1,2).$$

As in \cite{FMW1}, in view of the construction of solutions to the system, the key result is the following proposition:
\begin{proposition}\sl
\label{PropK1K2} One has,
\be
\label{estnormM1L}
\pal h^2K_{1,L}WK_{2,L}W^* \pal_{{\mathcal L}(C^0_b(I_L))} + \pal h^2K_{2,L}W^*K_{1,L}W \pal_{{\mathcal L}(C^0_b(I_L))}= \ord (h^{\frac1{3}}),
\ee
\be
\label{estnormK2L}
 |hK_{2,L}W^* v(0)| + |hK_{1,L}W v(0)| = \ord (\sup_{I_L}|v|),
\ee
uniformly with respect to $v\in C^0_b(I_L)$ and $h>0$ small enough.
\end{proposition} 
\begin{remark}\sl 
Actually, it can also be proved that $\pal hK_{2,L}W^* \pal_{{\mathcal L}(C^0_b(I_L))} = \ord (h^{-\frac16})$, but for our purpose the better estimate \eqref{estnormK2L} on $hK_{2,L}W^* v(0)$ is needed.
\end{remark}

\begin{proof} We denote by $U_1(x,t)$ (respectively $U_2(x,t)$) the distributional kernel of $h^2\W_{1,L}K_{1,L}W$ (resp. $h^2\W_{2,L}K_{2,L}W^*$). An integration by parts shows that we have,

\be
\label{noyaux}
\begin{aligned}
U_1(x,t)=& \widetilde U_1(x,t)+hr_1(0)u_{1,L}^-(x)u_{1,L}^+(0)\delta_{t=0};\\
U_2(x,t)=&  \widetilde U_2(x,t)+h{r_1(0)}u_{2,L}^-(x)u_{2,L}^+(0)\delta_{t=0},
\end{aligned}
\ee
with, 
$$
\begin{aligned}
& \widetilde U_1(x,t):= u_{1,L}^+(x)(W_1u_{1,L}^-)(t){\bf 1}_{\{t<x\}}+u_{1,L}^-(x)(W_1u_{1,L}^+)(t){\bf 1}_{\{x<t<0\}}\\
&  \widetilde U_2(x,t):=u_{2,L}^+(x)(W_2u_{2,L}^-)(t){\bf 1}_{\{t<x\}}+u_{2,L}^-(x)(W_2u_{2,L}^+)(t){\bf 1}_{\{x<t<0\}}.
\end{aligned}
$$
where
$W_1={}^tW=r_0-ihD_x \cdot r_1$, $W_2={}^tW^*=r_0-i r_1 hD_x$,
and we first prove the estimate on  $h^2K_{1,L}WK_{2,L}W^*$. For $v\in C^0_b(I_L)$, we have to study,
\be
\label{K1K2}
\begin{aligned}
h^2K_{1,L}WK_{2,L}W^*v(x) & =\frac1{h^2\W_{1,L}\W_{2,L}} \iint U_1(x,t)U_2(t,s) v(s) ds dt\\
& = \cO(h^{-\frac23})\iint U_1(x,t)U_2(t,s) v(s) ds dt.
\end{aligned}
\ee

We divide $\iint U_1(x,t)U_2(t,s) v(s) ds dt$ into four terms,

\be
\label{K1K2bis}
\iint U_1(x,t)U_2(t,s) v(s) ds dt=A_1(x)+A_2(x)+A_3(x)+A_4(x)
\ee
with,
\be
\begin{aligned}
& A_1(x):=\iint \widetilde U_1(x,t)\widetilde U_2(t,s) v(s) ds dt;\\
& A_2(x):= h{r_1(0)}u_{2,L}^+(0)v(0)\int_{-\infty}^{0}\widetilde U_1(x,t) u_{2,L}^-(t)dt;\\
& A_3(x):= hr_1(0)u_{1,L}^-(x)u_{1,L}^+(0)\int_{-\infty}^{0}\widetilde U_2(0,s) v(s)ds;\\
& A_4(x) := h^2|r_1(0)|^2u_{1,L}^-(x)u_{1,L}^+(0)u_{2,L}^-(0)u_{2,L}^+(0)v(0).
\end{aligned}
\ee

By construction, we have $u_{j,L}^-(x)u_{j,L}^+(t)=\cO(1)$ uniformly on $\{ x\leq t\leq 0\}$, and therefore,
\be
\label{estA4}
\sup_{I_L}|A_4| =\cO(h^2)\sup_{I_L}|v|.
\ee

For the same reason,

\be
\label{estA3.1}
\sup_{I_L}|A_3| =\cO(h)\left| \int_{-\infty}^{0}\widetilde U_2(0,s) v(s)ds\right|,
\ee

and, denoting by $b(E)\in\C$ the unique solution of $V_2(x)=E$ close to $b$, we write,
$$
\int_{-\infty}^{0}\widetilde U_2(0,s) v(s)ds=\int_{-\infty}^{\re b(E)}\widetilde U_2(0,s) v(s)ds+\int_{\re b(E)}^{0}\widetilde U_2(0,s) v(s)ds.
$$

By the same computations as in \cite{FMW1}, Section 3, we have (see \cite{FMW1}, Formula (3.11)),
$$
\int_{-\infty}^{\re b(E)}\widetilde U_2(0,s) v(s)ds =\cO(h^{\frac23})\sup_{I_L}|v|.
$$
On the other hand, on the interval $[\re b(E), 0]$ the functions $u_{2,L}^\pm$ are oscillating, and  we have (see, e.g., \cite{FMW1}, Section 8),
$$
|u_{2,L}^\pm (s)| + |h(u_{2,L}^\pm)' (s)| =\cO(h^{\frac16} |s-b(E)|^{-\frac14}). 
$$
As a consequence,
$$
\int_{\re b(E)}^{0}\widetilde U_2(0,s) v(s)ds =\cO(h^{\frac13})\sup_{I_L}|v|,
$$
and thus,
\be
\int_{-\infty}^{0}\widetilde U_2(0,s) v(s)ds =\cO(h^{\frac13})\sup_{I_L}|v|.
\ee
Inserting into (\ref{estA3.1}), we obtain,
\be
\label{estA3}
\sup_{I_L}|A_3| =\cO(h^{\frac43})\sup_{I_L}|v|.
\ee

Concerning $A_2$, since
$|u_{2,L}^+(0)|=\cO(h^{\frac16})$ uniformly, we have,
\be
\label{estA2prov}
A_2(x) =\cO(h^{\frac76})\left|\int_{-\infty}^{0}\widetilde U_1(x,t) u_{2,L}^-(t)dt\right|\sup_{I_L}|v|.
\ee
We prove,
\begin{lemma}\sl
\label{lemsigma}
$$
\int_{-\infty}^{0}\widetilde U_1(x,t) u_{2,L}^-(t)dt=\cO(h^{\frac{1}{3}}).
$$
\end{lemma}
\begin{proof} Using that $\widetilde U_1(x,t)=\cO(1)$ uniformly on $(-\infty, \re c(E)]$, and that, for any $\delta >0$, there exists $\alpha >0$ constant such that,
$$
\begin{aligned}
& u_{2,L}^-(t)=\cO(h^{\frac16}e^{-\alpha |t|/h}) \mbox{ in } (-\infty, \re b(E)-\delta];\\
& u_{2,L}^-(t)=\cO(h^{\frac16}|b(E)-t|^{-\frac14}e^{-\alpha |b(E)-t|^{\frac32}/h}) \mbox{ in } [\re b(E)-\delta, \re b(E)-h^{\frac23}];\\
& u_{2,L}^-(t)=\cO(1) \mbox{ in }   (-\infty, \re c(E)],
\end{aligned}
$$
we immediately obtain,
$$
\int_{-\infty}^{\re b(E)}\widetilde U_1(x,t) u_{2,L}^-(t)dt=\cO(h^{\frac23}).
$$
On the other hand, when $t\in [\re b(E), 0]$, we have $\widetilde U_1(x,t)=\cO(h^{1/6})$ and $u_{2,L}^-(t)=\cO(h^{1/6}|t-b(E)|^{-1/4})$.
Therefore,
$$
\int^{0}_{\re b(E)}\widetilde U_1(x,t) u_{2,L}^-(t)dt=\cO(h^{\frac13}),
$$
and the result follows.
\end{proof}

Going back to (\ref{estA2prov}), this gives us,
\be
\label{estA2}
\sup_{I_L}|A_2| =\cO(h^{\frac{3}{2}})\sup_{I_L}|v|.
\ee

Concerning $A_1$, we first observe that, by the same estimates as in \cite{FMW1}, Section 3 (see \cite{FMW1}, proof of Proposition 3.1), we have,
$$
\int_{-\infty}^{\re b(E)}\int_{-\infty}^{\re b(E)}\left|\widetilde U_1(x,t)\widetilde U_2(t,s)\right|  ds dt =\cO(h^{\frac43})
$$
uniformly for $x\in (-\infty, 0]$ and $h>0$ small enough,
and thus,
\be
\label{finalestA10}
\sup_{x\in I_L}\left|\int_{-\infty}^{\re b(E)}\int_{-\infty}^{\re b(E)}\widetilde U_1(x,t)\widetilde U_2(t,s) v(s) ds dt \right|=\cO(h^{\frac43})\sup_{I_L}|v|.
\ee
It remains to study the three quantities,
$$
\begin{aligned}
& A_{1,1}(x):= \int_{\re b(E)}^{0}dt\int_{\re b(E)}^{0}\widetilde U_1(x,t)\widetilde U_2(t,s) v(s) ds;\\
& A_{1,2}(x):=\int_{-\infty}^{\re b(E)}dt\int_{\re b(E)}^{0}\widetilde U_1(x,t)\widetilde U_2(t,s) v(s) ds ;\\
& A_{1,3}(x):=  \int_{\re b(E)}^{0}dt\int_{-\infty}^{\re b(E)}\widetilde U_1(x,t)\widetilde U_2(t,s) v(s) ds.
\end{aligned}
$$

For $A_{1,2}(x)$, since $t\leq s$ on the domain of integration, we have,
$$
A_{1,2}(x) =\int_{-\infty}^{\re b(E)} \widetilde U_1(x,t)u_{2,L}^-(t)dt \int_{\re b(E)}^{0}(W_2u_{2,L}^+)(s)v(s)ds,
$$
and thus, since $W_2u_{2,L}^+$ is $\cO(h^{1/6}|s-b(E)|^{-1/4})$ on $[\re b(E),0]$,
\be
\label{estA22-1}
A_{1,2}(x)=\cO(h^{1/6})\sup_{I_L}|v|\int_{-\infty}^{\re b(E)} \widetilde U_1(x,t)u_{2,L}^-(t)dt.
\ee
Thanks to the exponential decay of $u_{2,L}^-$ on $(-\infty, \re b(E))$, we immediately see that, for any $\delta >0$, we have,
\be
\label{estA22-2}
\int_{-\infty}^{\re b(E)-\delta} \widetilde U_1(x,t)u_{2,L}^-(t)dt =\cO(e^{-\alpha /h}),
\ee
with $\alpha =\alpha(\delta) >0$. On the other hand, if $\delta$ is sufficiently small, for $t\in [\re b(E)-\delta,\re b(E)]$, we have,
$$
\begin{aligned}
& \widetilde U_1(x,t)=\cO (h^{\frac16});\\
& u_{2,L}^-(t) =\cO(h^{\frac16}|t-b(E)|^{-\frac14}e^{-\beta (\re b(E) -t)^{3/2}/h}),
\end{aligned}
$$
where $\beta >0$ is some positive constant. Hence,
$$
\int_{\re b(E)-\delta}^{\re b(E)} \widetilde U_1(x,t)u_{2,L}^-(t)dt =\cO(h^{\frac13})\int_{-\delta}^\delta t^{-\frac14}e^{-\beta t^{3/2}/h}dt=\cO(h^{5/6}),
$$
and we finally obtain,
\be
\label{finalestA12}
\sup_{I_L} |A_{1,2}|=\cO(h)\sup_{I_L}|v|.
\ee

Concerning $A_{1,3}(x)$, we have $s\leq t$ on the domain of integration, and thus,
$$
A_{1,3}(x) =\int_{\re b(E)}^{0} \widetilde U_1(x,t)u_{2,L}^+(t)dt \int_{-\infty}^{\re b(E)}(W_2u_{2,L}^-)(s)v(s)ds.
$$
Since $\int_{-\infty}^{\re b(E)} |W_2u_{2,L}^-(s)|ds= \cO(h^{\frac23})$ (this can be seen, e.g., as in \cite{FMW1}, Section 3), we deduce,
\be
\label{estA23-1}
A_{1,3}(x)=\cO(h^{\frac23}\sup_{I_L}|v|)\int_{\re b(E)}^{0} \widetilde U_1(x,t)u_{2,L}^+(t)dt.
\ee
On the other hand, for $t\in [\re b(E),0]$, 
$$
\widetilde U_1(x,t)u_{1,L}^+(t)=\cO(h^{1/3}|t-b(E)|^{-1/4}),
$$
and thus
\be
\label{estfinaleA13}
A_{1,3}=\cO(h)\sup_{I_L}|v|.
\ee

Concerning $ A_{1,1}(x)$, we write,
\be
\label{decompA11}
A_{1,1}(x)= A_{1,1}^+(x)+A_{1,1}^-(x)
\ee
with,
$$
\begin{aligned}
& A_{1,1}^\pm(x):=  \int_{\re b(E)}^{0}\widetilde U_1(x,t)u_{2,L}^\pm(t)w_\pm (t)dt ;\\
& w_+(t):=\int_{\re b(E)}^t (W_2u_{2,L}^-)(s)v(s)ds;\\
& w_-(t):=\int_t^{0} (W_2u_{2,L}^+)(s)v(s)ds,
\end{aligned}
$$
and we observe that, since $|u_{2,L}^\pm|+|W_2u_{2,L}^\pm| =\cO(h^{\frac16} |E-V_2|^{-\frac14}$) 
and $(E-V_2)^{-\frac14}$ is integrable on $[\re b(E), 0]$, we have,
\be
\label{estw(t)}
w_\pm (t) =\cO(h^{\frac16})\sup_{I_L}|v|.
\ee
In addition, by definition, we also have,
\be
\label{estw'(t)}
w_\pm' (t) =\cO(\sup_{I_L}|v|).
\ee

Now, we fix $\lambda \geq1$  arbitrarily large, and we write,
\be
\begin{aligned}
A_{1,1}^\pm(x) =\int_{\re b(E)}^{\re b(E)+\lambda h^{2/3}} & \widetilde U_1(x,t)u_{2,L}^\pm(t)w_\pm (t)dt\\
&+\int_{\re b(E)+\lambda h^{2/3}}^{0}\widetilde U_1(x,t)u_{2,L}^\pm(t)w_\pm (t)dt,
\end{aligned}
\ee
and thus, using \eqref{estw(t)} and the fact that $\widetilde U_1(x,t) =\cO(h^{1/6})$ 
and $u_{2,L}^\pm(t)=\cO (h^{1/6}|t-b(E)|^{-1/4})$ 
when $t\in [\re b(E), \re b(E)+\lambda h^{2/3}]$, we obtain
\be
A_{1,1}^\pm(x) =\int_{\re b(E)+\lambda h^{2/3}}^{0}\widetilde U_1(x,t)u_{2,L}^\pm(t)w_\pm (t)dt +\cO(h)\sup_{I_L}|v|.
\ee

We first assume $x\leq \re b(E)$. In this case, one has $x\leq t$ on the domain of integration, and thus,
\be
\label{estA11-1}
A_{1,1}^\pm(x) =C(v)u_{1,L}^-(x) +\cO(h)\sup_{I_L}|v|.
\ee
with,
$$
C(v):=\int_{\re b(E)+\lambda h^{2/3}}^{0}W_1u_{1,L}^+(t)u_{2,L}^\pm(t)w_\pm (t)dt
$$
Now, on the interval $[\re b(E)+\lambda h^{2/3}, 0]$, since we stay far away from the turning points of $V_1$, we can use the WKB expansion of $u_{1,L}^+(t)$,
\be
\begin{aligned}
& u_{1,L}^-(t)=\frac{2h^{\frac16}}{\sqrt\pi}(E-V_1(t))^{-\frac14}\sin\left(h^{-1}\nu_1(t)+\frac{\pi}4\right)+\cO(h^{\frac76});\\
& h(u_{1,L}^-)'(t)=\frac{2h^{\frac16}}{\sqrt\pi}(E-V_1(t))^{\frac14}\cos\left(h^{-1}\nu_1(t)+\frac{\pi}4\right)+\cO(h^{\frac76}),
\end{aligned}
\ee
where we have set,
\be
\label{defnu1}
\nu_1(t):= \int_{a(E)}^t\sqrt{E-V_1(s)}ds.
\ee
On the other hand, concerning $u_{2,L}^-$, we have (see, e.g.,  \cite[Section 8]{FMW1}, and \cite{Ya}),
\be
\label{approxYaf}
u_{2,L}^-(t)=2(\xi_2'(t))^{-\frac12}\check \ai (h^{-\frac23}\xi_2(t))+\cO(h),
\ee

where $\ai$ is the usual Airy functions, $\check \ai(t):=\ai(-t)$, and where  $\xi_2=\xi_2(t;E)$ is the analytic continuation to complex values of $E$ of the function defined for $E$ real by,

$$
\begin{aligned}
& \xi_2(t;E) := \left( \frac32\int_{b(E)}^t\sqrt{E-V_2(s)}ds\right)^{\frac23} \quad \mbox{when } t\geq b(E);\\
& \xi_2(t;E) := -\left( \frac32\int_t^{b(E)}\sqrt{V_2(s)-E}ds\right)^{\frac23}\quad \mbox{when } t\leq b(E).
\end{aligned}
$$

By choosing $\lambda$ sufficiently large, we also see that $h^{-\frac23}\xi_2(t)$ becomes arbitrarily large when $t\in [\re b(E)+\lambda h^{2/3}, 0]$. Thus, there we can use the asymptotic behaviour of $\check \ai$ at infinity (see, e.g., \cite{Ol}),
$$
\check\ai (y)=\frac1{\sqrt \pi}y^{-\frac14}\sin\left( \frac23 y^{\frac32}+\frac{\pi}4\right)+\cO(|y|^{-\frac14-\frac32});\\
$$
valid uniformly as $|y|\to \infty$, $|\arg y| \leq \frac{2\pi}3-\delta$ ($\delta >0$ arbitrarily small). We obtain,
$$
u_{2,L}^-(t)=\frac{2h^{1/6}(\xi_2'(t))^{-\frac12}}{\sqrt{\pi}(\xi_2(t))^{\frac14}}\sin\left( \frac{2\xi_2(t)^{\frac32}}{3h} +\frac{\pi}4\right)+\cO\left(\frac{h^{7/6}}{|\xi_2(t)|^{\frac14+\frac32}}\right)+\cO(h).
$$
Then, using that $|\xi_2(E)|$ behaves like $|t-b(E)|$ on this interval, we obtain,
\be
\label{CC+c-}
C(v)=C_+(v) + C_-(v) +R(v)
\ee
where
$$
C_\pm(v)=\int_{\re b(E)+\lambda h^{2/3}}^{0} h^{\frac13}\frac{a_\pm (t)e^{\pm i\nu_1(t)/h}}{(t-b(E))^{1/4}}\sin\left( \frac{2\xi_2(t)^{\frac32}}{3h} +\frac{\pi}4\right)
w_\pm (t)dt
$$
with $a_{\pm}(t)$ smooth, and where,
$$
R(v) = \int_{\re b(E)+\lambda h^{2/3}}^{0}\left( \frac{\cO(h^{\frac76+\frac16})}{|t-b(E)|^{1/4}} +\frac{\cO(h^{\frac16+\frac76})}{|t-b(E)|^{\frac14+\frac32}}+\cO(h^{\frac76})\right)
|w_\pm (t)|dt.
$$
By the same arguments as before (in particular \eqref{estw(t)}), we obtain,
\be
\label{estR(t)}
R(v)=\cO(h^2+h+h^{\frac43})\sup_{I_L} |v|=\cO(h)\sup_{I_L} |v|.
\ee
On the other hand, setting,
\be
\label{defnu2}
 \nu_2(t):= \int_{b(E)}^t\sqrt{E-V_2(s)}ds,
\ee
we see that $C_+(v)$ and $C_-(v)$ are sums of terms of the type,
$$
B_+=\int_{\re b(E)+\lambda h^{2/3}}^{0} h^{\frac13}\frac{a (t)e^{\pm i(\nu_1(t)+ \nu_2(t))/h}}{(t-b(E))^{1/4}}
w_\pm (t)dt,
$$
or of the type,
$$
B_-=\int_{\re b(E)+\lambda h^{2/3}}^{0} h^{\frac13}\frac{a (t)e^{\pm i(\nu_1(t)- \nu_2(t))/h}}{(t-b(E))^{1/4}}
w_\pm (t)dt,
$$
with $a(t)$ smooth, and where the various $\pm$ are not related each other. 

Here we observe that, for $t\in[\re b(E)+\lambda h^{2/3}, 0]$, one has,
\begin{itemize}
\item $|\im \nu_1(t)| + |\im \nu_2(t)| =\cO(h)$;
\item $\re(\nu_1'(t)+\nu_2'(t))=\re(\sqrt {E-V_1(t)}+\sqrt{E-V_2(t)})\geq \frac1{C}$ for some constant $C>0$;
\item $\re(\nu_1'(t)-\nu_2'(t))=\re(\sqrt {E-V_1(t)}-\sqrt{E-V_2(t)})$ vanishes at $t=0$ only;
\item $|\re(\nu_1''(t)-\nu_2''(t))|\geq \frac1{C}$ for some constant $C>0$.
\end{itemize}

For the $B_+$-type terms, we write,
\be
\label{nonstation}
e^{\pm i(\nu_1(t)+ \nu_2(t))/h}=\frac{\pm h}{i(\nu_1'(t)+ \nu_2'(t))}\frac{d}{dt}(e^{\pm i(\nu_1(t)+ \nu_2(t))/h}),
\ee
and we make an integration by parts. Using the notation $\varphi:=\pm (\nu_1+\nu_2)$, we obtain,
$$
B_+= \cO(h^{\frac43})\sup_{I_L} |v| +ih^{\frac43}\int_{\re b(E)+\lambda h^{2/3}}^{0}e^{ i\varphi(t)/h}\frac{d}{dt}\left(\frac{ a(t)w_\pm(t)}{(t-b(E))^{\frac14}\varphi'(t)}\right)dt
$$
and thus, using \eqref{estw'(t)} and the fact that $\frac{d}{dt}\left(\frac{ a(t)}{(t-b(E))^{\frac14}\varphi'(t)}\right)$ is  $\cO(|t-b(E)|^{-\frac54})$,
\be
\label{estB+}
B_+= \cO(h^{\frac43})\sup_{I_L} |v|.
\ee
Concerning the $B_-$-type terms, in view of performing a stationary-phase argument, let $\chi\in C_0^\infty (\re b(E), 0)$ be a ($h$-independent) cut-off function, such that $\chi =1$ near 0. We write $B_-=B_{-,1}+B_{-,2}$, with,
$$
\begin{aligned}
& B_{-,1}:=\int_{\re b(E)+\lambda h^{2/3}}^{0} h^{\frac13}(1-\chi(t))\frac{a (t)e^{\pm i(\nu_1(t)- \nu_2(t))/h}}{(t-b(E))^{1/4}}
w_\pm (t)dt,
;\\
& B_{-,2}:=\int_{\re b(E)+\lambda h^{2/3}}^{0} h^{\frac13}\chi (t)\frac{a (t)e^{\pm i(\nu_1(t)- \nu_2(t))/h}}{(t-b(E))^{1/4}}
w_\pm (t)dt.
\end{aligned}
$$
Exactly as for $B_+$, we see,
\be
\label{estB-1}
B_{-,1}=\cO(h^{\frac43})\sup_{I_L} |v|.
\ee

In order to estimate $B_{-,2}$, we need some special version of the stationary-phase theorem (it is probably well-known, but we did not find any reference for it).
\begin{lemma}\sl 
\label{phasestat1}
Let  $\chi_0\in C_0^\infty (\R;[0,1])$ with $\chi_0=1$ near 0, and $\psi\in C^\infty(\R;\R)$ admitting 0 as unique stationary point in ${\rm Supp}\chi_0$ with $\psi''(0)\not=0$. Then, denoting by $K$ the convex hull of ${\rm Supp}\chi_0$ and by ${\rm sgn}\hskip 1pt\psi''(0)$ the sign of $\psi''(0)$, one has, for  $f\in C^2(\R)$,
\begin{equation}
\label{3.4}
\int e^{i\psi (t)/h}\chi_0(t)f(t)dt =f(0)e^{i\frac{\pi}4 \hskip 1pt {\rm sgn}\hskip 1pt\psi''(0)}\sqrt{\frac{2\pi h}{|\psi''(0)|}} + \cO(h)\sup_{K}(|f'|+|f''|),
\end{equation}
uniformly with respect to $h>0$ small enough.
\end{lemma}
\begin{proof} First of all, by a smooth change of variable (depending only on $\psi$), we can assume that $\psi =\pm \mu t^2/2$ with $\mu >0$ constant. Then, writing,
$$
f(t) =f(0)+tg(t)
$$
with,
$$
g(t):=\int_0^1 f'(\theta t)d\theta,
$$
we obtain,
$$
\int e^{i\psi/h}\chi_0fdt =f(0)\int e^{\pm i\mu t^2/2h}\chi_0dt \pm\frac{h}{i\mu } \int 
\frac{d}{dt}(e^{\pm i\mu t^2/2h})\chi_0gdt.
$$
By the standard stationary-phase theorem (see, e.g., \cite{Ma}), we have,
\be
\label{standstatphase}
\int e^{\pm i\mu t^2/2h}\chi_0(t)dt=\mu^{-\frac12}e^{\pm i\frac{\pi}4}\sqrt{2\pi h}+\cO(h^\infty),
\ee
 and thus, by an integration by parts, we obtain,
\be
\label{intpart}
\int e^{i\psi/h}\chi_0fdt =f(0)(\mu^{-\frac12}e^{\pm i\frac{\pi}4}\sqrt{2\pi h}+\cO(h^\infty)) \pm \frac{ih}{\mu}\int e^{\pm i\mu t^2/2h}\frac{d}{dt}(\chi_0g)dt,
\ee
where the $\cO(h^\infty)$ does not depend on $f$. Then, the result follows from the fact that $\sup_{{\rm Supp}\chi_0}(|g|+|g'|)\leq \sup_K(|f'|+|f''|)$.
\end{proof}
\begin{remark}
\label{rem3.5}
The formula \eq{3.4} stays valid when the integration is restricted to the half line $\R_-$ or $\R_+$ just by replacing the first term of the RHS with its half.
To see this, it is enough to check, instead of \eq{standstatphase},  that
$$
\int_{\R_\pm} e^{\pm i\mu t^2/2h}\chi_0(t)dt=\frac 12\mu^{-\frac12}e^{\pm i\frac{\pi}4}\sqrt{2\pi h}+\cO(h^\infty),
$$
(see also Lemma \ref{statphase2} for more general cases) and that the endpoint term $\frac{ih}\mu g(0)$ arising from the integration by parts is of $\ord(h)\sup_K|f'|$.
\end{remark}

\begin{lemma}\sl 
\label{phasestat2}
Let  $I\subset \R$ be an open interval containing 0, and $\psi\in C^\infty(\R;\R)$ admitting 0 as unique stationary point in $\bar I$ with $\psi''(0)\not=0$. Then, one has,
for $u\in C_0^1(I)$,
$$
\int e^{i\psi (t)/h}u (t)dt =\cO(\sqrt{h})\sup (|u|+|u'|),
$$
uniformly with respect to $h>0$ small enough.
\end{lemma}
\begin{proof} 
Setting,
$$
f(t):= \frac1{\sqrt{2\pi h}}\int e^{-(t-s)^2/2h}u(s)ds,
$$
we have,
$$
f(t)-u(t)=\frac1{\sqrt{2\pi h}}\int e^{-(t-s)^2/2h}(u(s)-u(t))ds,
$$
and thus, since $|u(s)-u(t)| \leq |s-t|\sup |u'|$ and $\int e^{-(s-t)^2/2h}|s-t|ds = 2h$,
$$
|f(t)-u(t)|\leq\sqrt{\frac{2h}{\pi}} \hskip 0.1cm \sup |u'|.
$$
As a consequence, fixing $\chi_0\in C_0^\infty (\R;[0,1])$ with $\chi_0=1$ on $I$ (so that $\chi_0 u =u$), we have,
\be
\label{estf1}
\left|\int e^{i\psi (t)/h}u (t)dt - \int e^{i\psi (t)/h}\chi_0(t)f (t)dt\right| \leq  |I|\sqrt{\frac{2h}{\pi}} \hskip 0.1cm \sup |u'|.
\ee
On the other hand, applying Lemma \ref{phasestat1} to $f$, we have,
\be
\label{estf2}
\int e^{i\psi (t)/h}\chi_0(t)f (t)dt =f(0)\sqrt{\frac{2\pi h}{|\psi''(0)|}}  + \cO(h)\sup(|f'|+|f''|).
\ee
Now, we can write,
$$
\begin{aligned}
 & f'(t) =\frac1{\sqrt{2\pi h}}\int e^{-(t-s)^2/2h}u'(s)ds;\\
 & f''(t) = \frac{-h^{-1}}{\sqrt{2\pi h}}\int e^{-(t-s)^2/2h}(t-s)u'(s)ds.
\end{aligned}
$$
Therefore,
$$
|f'(t)| \leq \sup |u'|,
$$
and,
$$
|f''(t)| \leq \frac2{\sqrt{2\pi h}}\sup |u'|.
$$
Since also $|f(0)|\leq \sup|u|$, the result follows form (\ref{estf1})-(\ref{estf2}).
\end{proof}
\begin{remark}
\label{rem3.7}
As in Remark \ref{rem3.5}, this lemma remains true when the integration is restricted to a half line $\R_\pm$.
\end{remark}

Applying the previous lemma and remark with 
$$\psi = \pm (\nu_1- \nu_2),\quad u (t)=h^{\frac13}\frac{a (t)}{(t-b(E))^{1/4}}
w_\pm (t),
$$ we obtain,
$$
B_{-,2}=\cO(h^{\frac56})\sup_{{\rm Supp \chi}} (| w_\pm|+|w_\pm'|),
$$
and thus, by (\ref{estw(t)}) and the fact that, on ${\rm Supp}\chi$,  $w_\pm'(t)=\cO(h^{\frac16})v(t)$,
\be
\label{estB-2}
B_{-,2}=\cO(h)\sup_{I_L}|v|.
\ee
Gathering (\ref{estA11-1}), (\ref{CC+c-}), \eqref{estR(t)}, \eqref{estB+}, \eqref{estB-1} and \eqref{estB-2}, and using the fact that $u_{1,L}^-$ is uniformly bounded on $I_L$, we obtain,
\be
\sup_{(-\infty, \re b(E)]} |A_{1,1}^\pm| =\cO(h)\sup_{I_L}|v|.
\ee

Now, when $x\in[\re b(E), 0]$, the only difference is that $A_{1,1}^\pm(x)$ must be written as,
$$
\begin{aligned}
A_{1,1}^\pm(x)=u_{1,L}^+(x)\int_{\re b(E)}^x(W_1u_{1,L}^-)(t)u_{2,L}^\pm(t)w_\pm (t)dt\\
 + u_{1,L}^-(x)\int_x^{0}(W_1u_{1,L}^+)(t)u_{2,L}^\pm(t)w_\pm (t)dt,
\end{aligned}
$$
but since $u_{1,L}^+$ is uniformly bounded on $[\re b(E), 0]$, the previous argument works again for each of the two terms, as long as $x$ remains away from the critical point of $\nu_1-\nu_2$ (in this case, it suffices to choose $\chi$ in such a way that $x\notin {\rm Supp} \chi$). When $x$ is closed to 0, there are two changes in the proof of Lemma \ref{phasestat1}. The first one is that an extra-term appears in (\ref{intpart}), the boundary term $\mp ih\mu^{-1}e^{\pm i\mu \tilde x^2/2h}\chi_0(\tilde x)g(\tilde x) $ (where $\tilde x$ is the value of $x$ after the change of variable that transforms $\psi$ into $\pm \mu t^2/2$), but this term is clearly $\cO(h)\sup_K|f'|$. The other change concerns (\ref{standstatphase}), since we now have to estimate $\int_{\pm t\geq \pm\tilde x} e^{\pm i\mu t^2/2h}\chi_0(t)dt$. We need to prove,

\begin{lemma}\sl 
\label{statphase2}
For any $\mu>0$ independent of $h$, $a\in\R$ (which may depend on $h$) and $\chi_0$ as in Lemma \ref{phasestat1}, one has,
$$
\int_{\pm t\geq \pm a} e^{\pm i\mu t^2/2h}\chi_0(t)dt=\cO(\sqrt{h}).
$$
uniformly with respect to $h>0$ small enough.
\end{lemma}
\begin{proof} We treat the case $t\geq a$ only (the other one being similar), and we can assume that $a\in {\rm Supp} \,\chi_0$ (otherwise we already know that the integral is $\cO(\sqrt{h})$). We fix some $\sigma >0$ and, if $a\geq -h^\sigma$, we write,
$$
\begin{aligned}
\int_{ t\geq a} e^{\pm i\mu t^2/2h}\chi_0(t)dt & =\int_a^{a+2h^\sigma} e^{\pm i\mu t^2/2h}\chi_0(t)dt+\int_{ t\geq a+2h^\sigma} e^{\pm i\mu t^2/2h}\chi_0(t)dt\\
& = \cO(h^\sigma) \pm \frac{h}{i\mu }\int_{ t\geq a+2h^\sigma} \frac{d}{dt}\left( e^{\pm i\mu t^2/2h}\right)\frac{\chi_0(t)}{t}dt\\
& = \cO(h^\sigma +h^{1-\sigma}) \pm \frac{ih}{\mu }\int_{ t\geq a+2h^\sigma} e^{\pm i\mu t^2/2h}\frac{d}{dt}\left( \frac{\chi_0(t)}{t}\right)dt\\
& =  \cO(h^\sigma+h^{1-\sigma}) +\cO(h)\int_{ a+2h^\sigma}^{+\infty} \frac{dt}{t^2}\\
& =  \cO(h^\sigma+h^{1-\sigma}).
\end{aligned}
$$
If $a\leq-h^{\sigma}$, we write,
$$
\begin{aligned}
\int_{ t\geq a} e^{\pm i\mu t^2/2h}\chi_0(t)dt & =\int_a^{-h^\sigma} e^{\pm i\mu t^2/2h}\chi_0(t)dt+\int_{ t\geq -h^\sigma} e^{\pm i\mu t^2/2h}\chi_0(t)dt\\
& =\int_a^{-h^\sigma} e^{\pm i\mu t^2/2h}\chi_0(t)dt +\cO(h^\sigma+h^{1-\sigma})\\
& = \pm \frac{h}{i\mu }\int_a^{-h^\sigma} \frac{d}{dt}\left( e^{\pm i\mu t^2/2h}\right)\frac{\chi_0(t)}{t}dt+\cO(h^\sigma+h^{1-\sigma})\\
& =\cO(h^\sigma+h^{1-\sigma}),
\end{aligned}
$$
where the last estimates comes again from an integration by parts. Taking $\sigma =\frac12$, the result follows.
\end{proof}
 
Hence, the estimate remains the same as in Lemma \ref{phasestat2}, and we finally obtain,
\be
\label{finalestA11}
\sup_{I_L} |A_{1,1}|=\cO(h)\sup_{I_L}|v|.
\ee

Then, the required estimate on the norm of $h^2K_{1,L}WK_{2,L}W^*$ follows from \eqref{K1K2},  \eqref{K1K2bis},  \eqref{estA4}, \eqref{estA3},  \eqref{estA2},  \eqref{finalestA10},  \eqref{finalestA12},  \eqref{estfinaleA13} and  \eqref{finalestA11}. The same arguments also apply to $h^2K_{2,L}W^*K_{1,L}W$, and this proves \eqref{estnormM1L}.

Concerning \eqref{estnormK2L},  by using \eqref{noyaux} we have,
$$
\begin{aligned}
hK_{2,L}W^*v(0)  =\cO(h^{-\frac13}) \int_{-\infty}^{0} & \widetilde U_2(0,t) v(t) dt\\
&+\cO(h^{\frac23}){r_1(0)}u_{2,L}^-(0)u_{2,L}^+(0)v(0),
\end{aligned}
$$
and thus, using the fact that 
$$
|u_{2,L}^\pm (t)| + |W_2u_{2,L}^\pm (t)| =\cO(h^{\frac16})|t-b(E)|^{-\frac14}e^{- \frac1{h}\int^{b(E)}_t \sqrt{(V_2(s)-E)_+}ds},
$$
we obtain,
$$
\begin{aligned}
hK_{2,L}W^*v(0)  &=\cO(h^{-\frac13}) \int_{-\infty}^{0}  \widetilde U_2(0,t) v(t) dt +\cO(h)\sup_{I_L}|v|\\
&= \cO(h^{-\frac16})\int_{-\infty}^0W_2u_{2,L}^-(t)v(t) dt+\cO(h)\sup_{I_L}|v| \\
&= \cO(\sup_{I_L}|v|).
\end{aligned}
$$
Similar arguments hold for $hK_{1,L}Wv(0)$,  and Proposition \ref{PropK1K2} follows.
\end{proof}

 \subsection{Fundamental solutions on $I_R^\theta$}

Exactly as in \cite{FMW1}, the same constructions hold on $I_R^\theta$, and lead to fundamental solutions
$$K_{j,R}\, :\, C^0_b(I_R^\theta) \to C^2_b(I_R^\theta)\qquad (j=1,2),$$
of $P_j - E$ on $I_R^\theta$. Then, the same arguments as in the previous section also give,

\begin{proposition}\sl
\label{PropK1RK2R} One has,
\be
\label{estnormM1R}
\pal h^2K_{2,R}W^*K_{1,R}W \pal_{{\mathcal L}(C^0_b(I_R^\theta))} +\pal h^2K_{1,R}WK_{2,R}W^* \pal_{{\mathcal L}(C^0_b(I_R^\theta))} = \ord (h^{\frac1{3}}),
\ee
\be
\label{estnormK2R}
|hK_{1,R}W v(0)|+|hK_{2,R}W^* v(0)| = \ord (\sup_{I_R^\theta}|v|),
\ee
uniformly with respect to $v\in C_b^0(I_R^\theta)$ and $h>0$ small enough.
\end{proposition} 

\subsection{Solutions of the system on $I_L$ and $I_R^\theta$}

Following \cite[Section 4]{FMW1}, we set,
$$\begin{aligned}
& M_L:= h^2K_{1,L}WK_{2,L}W^* \quad ; \quad \widetilde M_L:=h^2K_{2,L}W^*K_{1,L}W\\
& M_R:=h^2K_{2,R}W^*K_{1,R}W \quad ; \quad \widetilde M_R:=h^2K_{1,R}WK_{2,R}W^*
\end{aligned}
$$
and, thanks to Propositions \ref{PropK1K2} and \ref{PropK1RK2R}, we see that the convergent series given by,
\be
\label{wL}
\begin{aligned}
& w_{1,L}:=\left(\begin{array}{c} \sum_{j\geq 0}M_L^ju_{1,L}^-\\-hK_{2,L} W^*\sum_{j\geq 0}M_L^ju_{1,L}^-\end{array}\right);\\
& w_{2,L}:=\left(\begin{array}{c} -hK_{1,L}W\sum_{j\geq 0}\widetilde M_L^j u_{2,L}^-\\ \sum_{j\geq 0}\widetilde M_L^ju_{2,L}^-,\end{array}\right)
\end{aligned}
\ee
are solutions to \eqref{sch} on $I_L$, while the convergent series,
\be
\label{wR}
\begin{aligned}
& w_{1,R}:=\left(\begin{array}{c} \sum_{j\geq 0}\widetilde M_R^ju_{1,R}^-\\
-hK_{2,R} W^*\sum_{j\geq 0}\widetilde M_R^ju_{1,R}^-
\end{array}\right);\\
& w_{2,R}:=\left(\begin{array}{c} 
-hK_{1,R} W\sum_{j\geq 0}M_R^ju_{2,R}^-\\ 
\sum_{j\geq 0}M_R^ju_{2,R}^-
\end{array}\right)
\end{aligned}
\ee
(where $u_{j,R}^-$ are the solutions to $(P_j-E)u_{j,R}^-=0$ constructed in \cite[Appendix 2]{FMW1}) 
are solutions to \eqref{sch} on $I_R^\theta$. In addition, one also has (see \cite[Proposition 4.1]{FMW1}),
\be
w_{j,L} \in L^2(I_L)\oplus L^2(I_L)\quad ; \quad w_{j,R} \in L^2(I_R^\theta)\oplus L^2(I_R^\theta).
\ee

\section{Existence and location of resonances}

The four solutions constructed in the previous section permits us to write the quantization condition that determines the resonances of $P$ in ${\mathcal D}_h(\delta_0,C_0)$ as,
\be
\label{condquant}
 {\mathcal W}_0(E)=0,
\ee
where ${\mathcal W}_0(E):={\mathcal W}(w_{1,L}, w_{2,L}, w_{1,R}, w_{2,R})$ stands for the Wronskian of $w_{1,L}$, $w_{2,L}$, $w_{1,R}$ and $w_{2,R}$.
Since this Wronskian is constant with respect to $x$, we plan to compute it at $x=0$. We first show,
\begin{proposition}\sl
\label{w(0)}
For $S=L,R$, one has,
$$
\begin{aligned}
& w_{1,S}(0) = \left(\begin{array}{c} u_{1,S}^-(0)\\0\end{array}\right)+\cO(h^{\frac13})\quad ; \quad w_{1,S}'(0)= \left(\begin{array}{c} (u_{1,S}^-)'(0)\\0\end{array}\right)+\cO(h^{-\frac23}); \\
& w_{2,S}(0) = \left(\begin{array}{c} 0 \\ u_{2,S}^-(0)\end{array}\right)+\cO(h^{\frac13})\quad ; \quad w_{2,S}'(0)= \left(\begin{array}{c} 0 \\ (u_{2,S}^-)'(0)\end{array}\right)+\cO(h^{-\frac23}).
\end{aligned}
$$
\end{proposition}
\begin{proof}
We write the proof for $j=1$ and $S=L$, the ones for $j=2$ or $S=R$ being similar. By Proposition \ref{PropK1K2} and \eqref{wL}, we have,
\be
\label{w1lapprox1}
w_{1,L}(0) = \left(\begin{array}{c} u_{1,L}^-(0)\\ -hK_{2,L}W^*u_{1,L}^-(0)\end{array}\right)+\cO(h^{\frac13}).
\ee
On the other hand, by \eqref{noyaux} and \eqref{wronsL}-\eqref{eq1},
$$
\begin{aligned}
hK_{2,L}W^*u_{1,L}^-(0)  =\cO(h^{-\frac13}) \int_{-\infty}^{0} & \widetilde U_2(0,t) u_{1,L}^-(t) dt\\
&+\cO(h^{\frac23})\overline{r_1(0)}u_{2,L}^-(0)u_{2,L}^+(0)u_{1,L}^-(0),
\end{aligned}
$$
and thus, using the fact that $|u_{j,L}^\pm (0)| =\cO(h^{\frac16})$,
\be
\label{K2Lu1L}
hK_{2,L}W^*u_{1,L}^-(0)  =\cO(h^{-\frac13}) \int_{-\infty}^{0}  \widetilde U_2(0,t) u_{1,L}^-(t) dt +\cO(h^{\frac76}).
\ee
Now, we show,
\begin{lemma}\sl 
\label{lemU2uL}
$$
\int_{-\infty}^{0}  \widetilde U_2(0,t) u_{1,L}^-(t) dt  =\cO(h).
$$
\end{lemma}
\begin{proof} Since $W_2u_{2,L}^-(t)$ is exponentially small on $(-\infty, \re b(E)-\delta]$ ($\delta >0$ arbitrarily small), 
both with respect to $t$ and $h$, we have,
\be
\label{U2u1L}
\begin{aligned}
\int_{-\infty}^{0}  \widetilde U_2(0,t) u_{1,L}^-(t) dt & =u_{2,L}^+(0)\int_{-\infty}^0W_2u_{2,L}^-(t)u_{1,L}^-(t) dt \\
& = \cO(h^{\frac16})\int_{\re b(E)-\delta}^0W_2u_{2,L}^-(t)u_{1,L}^-(t) dt +\cO(e^{-\alpha /h})
\end{aligned}
\ee
with $\alpha =\alpha(\delta)>0$. Then, using the asymptotic behaviour of $u_{2,L}^-$ and $u_{1,L}^-$ near $b(E)$ (see, e.g., \cite[Appendix 2]{FMW1}), we have (for some $\beta >0$ constant),
\be
\begin{aligned}
\int_{\re b(E)-\delta}^{\re b(E)} W_2u_{2,L}^-(t)u_{1,L}^-(t) dt 
 &= \cO(h^{\frac13})\int_{\re b(E)-\delta}^{\re b(E)} \frac{e^{-\beta |t-b(E)|^{3/2}/h}}{|t-b(E)|^{1/4}}dt\\
 &=\cO(h^{5/6}),
\end{aligned}
\ee
and, for any fixed $\lambda>0$ arbitrarily large,
\be
\begin{aligned}
\int_{\re b(E)}^{\re b(E)+\lambda h^{\frac23}}W_2u_{2,L}^-(t)u_{1,L}^-(t) dt 
 &= \cO(h^{\frac16})\int_{\re b(E)}^{\re b(E)+\lambda h^{\frac23}} dt\\
 &=\cO(h^{5/6}).
\end{aligned}
\ee
Moreover, if $\lambda$ has been taken large enough, we also have (see, e.g., \cite[Propositions A.5]{FMW1}),
\begin{align*}
 &\ \int_{\re b(E)+\lambda h^{\frac23}}^{0}W_2u_{2,L}^-(t)u_{1,L}^-(t) dt\\ 
=&\ \sum_{\sigma\in \{\pm 1\}^2}\int_{\re b(E)+\lambda h^{\frac23}}^{0}\frac{h^{\frac13}f_\sigma (t)}{(t-b(E))^{\frac14}} 
e^{i \nu_\sigma(t)/h}dt +\cO (h^{\frac76}),
\end{align*}
where $f_\sigma$ is smooth and bounded (together with all its derivatives), and where, for $\sigma=(\sigma_1,\sigma_2)\in \{\pm1\}^2$ we have used the notation,
$$
\nu_\sigma:= \sigma_1\nu_1+\sigma_2\nu_2.
$$
(Here, $\nu_1$ and $\nu_2$ are the functions defined in \eqref{defnu1} and \eqref{defnu2}.)
In particular, we see that $\nu_\sigma$ has no critical point in $[\re b(E)+\lambda h^{\frac23}, -\delta]$, 
and thus, writing $e^{i \nu_\sigma(t)/h}=\frac{h}{i\nu_\sigma'(t)}\frac{d}{dt}e^{i \nu_\sigma(t)/h}$, and integrating by parts, we obtain,
\be
\begin{aligned}
\int_{\re b(E)+\lambda h^{\frac23}}^{-\delta}W_2u_{2,L}^-(t)u_{1,L}^-(t) dt 
 &=\cO(h^{\frac76})+\cO(h^{\frac43})\int_{\re b(E)+\lambda h^{\frac23}}^{-\delta}(t-b(E))^{-\frac54}dt\\
 &=\cO(h^{\frac76}).
\end{aligned}
\ee
 It remains us to study $\int_{-\delta}^0W_2u_{2,L}^-(t)u_{1,L}^-(t) dt$. Since the phase functions $\nu_\sigma$ corresponding to $\sigma =\pm(1,-1)$ have a  non-degenerate critical point at 0, after a change of variables the corresponding integrals can be transformed into,
$$
h^{\frac13}\int_0^{\delta'}f(t)e^{it^2/2h}dt,
$$
with $f$ smooth and bounded together with all its derivatives, and $\delta'>0$ constant. But then, Lemma \ref{statphase2} can be applied and leads to,
$$
\left|\int_{-\delta}^0W_2u_{2,L}^-(t)u_{1,L}^-(t) dt \right| =\cO(h^{\frac56}),
$$
and by \eqref{U2u1L} Lemma \ref{lemU2uL} follows.
\end{proof}

Going back to \eqref{w1lapprox1} and \eqref{K2Lu1L}, the previous lemma implies,
\be
\label{w1lapprox2}
w_{1,L}(0) = \left(\begin{array}{c} u_{1,L}^-(0)\\ 0\end{array}\right)+\cO(h^{\frac13}).
\ee
Concerning $w_{1,L}'(0)$, we just observe that for $j=1,2$, the function $hD_xu_{j,L}^\pm$ have a behaviour of the same type as $u_{j,L}^\pm$, and the same computations as in the proof of Proposition  \ref{PropK1K2} also give,
\begin{align*}
\pal h^2D_xK_{1,L}WK_{2,L}W^* \pal_{{\mathcal L}(C^0_b(I_L))} &= \ord (h^{-\frac{2}{3}});\\
  hD_xK_{2,L}W^* v(0) &= \ord (h^{-1})\sup_{I_L}|v|.
 \end{align*}
(Observe that, when differentiating $K_{j,L}v(x)$, the terms that come out when the derivative acts on the $x$ of $\int_x$ or $\int^x$ cancel each other.) 

For the same reason, we also obtain,
$$
hD_xK_{2,L}W^*u_{1,L}^-(0)=\cO(h^{-\frac13}).
$$
Thus, in the same way as for $w_{1,L}(0)$ (that is, using \eqref{wL}), we have,
\be
\label{w1lapprox3}
w_{1,L}'(0) = \left(\begin{array}{c} (u_{1,L}^-)'(0)\\ 0\end{array}\right)+\cO(h^{-\frac23}).
\ee
Bu using \eqref{wL}-\eqref{wR} and Proposition \ref{estnormM1L}, we can repeat the same arguments for $j=2$ and/or $S=R$, and Proposition \ref{w(0)} follows.
\end{proof}

\begin{proposition}\sl
\label{approxW0}
For any $E\in {\mathcal D}_h(\delta_0,C_0)$, one has,
$$
{\mathcal W}_0(E)= \frac{4\sqrt{2}}{\pi^2} e^{-i\frac{\pi}4}h^{-\frac43}\cos \frac{{\mathcal A}(E)}{h} + {\mathcal O}(h^{-\frac76}),
$$
uniformly as $h\to 0_+$.
\end{proposition}
\begin{proof} 
Using Proposition \ref{w(0)}, and the fact that, for $S=L,R$ and $j=1,2$, one has $u_{j,S}(0)=\cO(h^{\frac16})$ and $(u_{j,S})'(0)=\cO(h^{-\frac56})$, we immediately obtain,
$$
{\mathcal W}_0(E)={\mathcal W}(u_{1,L}^-,u_{1,R}^-){\mathcal W}(u_{2,L}^-,u_{2,R}^-)+\cO(h^{-\frac76}).
$$
On the other hand, we know from standard WKB constructions (see also \cite[Appendix]{FMW1}) that we have,
\be
\begin{aligned}
& {\mathcal W}(u_{1,L}^-,u_{1,R}^-)= -\frac4{\pi}h^{-\frac23}\cos \frac{{\mathcal A}(E)}{h} + \cO(h^{\frac13});\\
& {\mathcal W}(u_{2,L}^-,u_{2,R}^-)=\frac{i\sqrt 2}{\pi}h^{-\frac23}e^{i\frac{\pi}4}+\cO(h^{\frac13}),
\end{aligned}
\ee
and the result follows.
\end{proof}

Now, we are able to establish the existence of resonances, together with a preliminary (but fundamental) result on their location. With the definition of $e_k(h)$ given  in \eqref{defekh}, we have,

\begin{theorem} \sl
\label{ThapproxRes}
Under Assumptions (A1)-(A5), there exists $\delta_0>0$ such that for any $C_0>0$, one has, for $h>0$ small enough
$$
{\rm Res}\,(P)\cap {\mathcal D}_h(\delta_0, C_0)
=\{E_k(h); k\in\Z\}\cap{\mathcal D}_h(\delta_0, C_0)
$$
where the
 $E_k(h)$'s are complex numbers that satisfy,
\be
\label{Ekprelim}
E_k(h)=e_k(h) + \cO(h^{\frac76}),
\ee
uniformly as $h\to 0_+$.
\end{theorem}
\begin{rem}
The ellipticity of $W$ assumed in Assumption (A5) is not used in this theorem. 
It will be used in the next section in the microlocal method for a more precise estimate.
\end{rem}
\begin{proof} 
We set,
$$
f(E,h):=\cos \frac{{\mathcal A}(E)}{h}-\frac{e^{i\frac{\pi}4}\pi^2h^{\frac43}}{4\sqrt 2}{\mathcal W}_0(E).
$$
Then, by Proposition \ref{approxW0}, we have $f(E,h)=\cO(h^{1/6})$, and the quantization condition \eqref{condquant} can be written as,
\be
\label{condquantbis}
\cos \frac{{\mathcal A}(E)}{h}=f(E,h).
\ee

We first observe that, near $E_0$, the roots of the equation $\cos ({\mathcal A}(E)/h)=0$ are precisely given by $E=e_k(h)$. Moreover, since ${\mathcal A}'(E_0)=\frac12\int_{a}^{c} (E_0-V_1(x))^{-1/2}dx\not=0$, we see that the distance between two consecutive $e_k(h)$'s is of order $h$. 
As a consequence, if $E$ is at a distance $\varepsilon h$ from the $e_k(h)$ with small enough positive $\varepsilon$, 
then $\cos ({\mathcal A}(E)/h)$ remains at some $h$-independent positive distance from $0$.
Therefore, we can apply the Rouch\'e theorem and conclude that, for each $k$ such that $e_k(h)\in [E_0-\delta, E_0+\delta]$, and for $h>0$ small enough, there exists a unique solution to \eqref{condquantbis} such that,
$$
E_k(h) = e_k(h) + o(h),
$$
and conversely, all the roots of \eqref{condquantbis} in ${\mathcal D}_h(\delta_0, C_0)$ are of this type. But since $f=\cO(h^{1/6})$, 
we immediately see that these roots actually satisfy,
$$
E_k(h) = e_k(h) +\cO(h^{\frac76}),
$$
and the result is proved.
\end{proof}

Now, in order to specify better the location of the resonances (in particular their widths), we will compare the corresponding resonant states with formal constructions that will be made microlocally (that is, in phase-space) near the characteristic set of $P-E_0$,
${\rm Char} (P-E_0)=\Gamma_1\cup \Gamma_2$ (see Figure 2).

\section{Microlocal constructions near the crossing points of the characteristic sets}

For $E\in {\mathcal D}_h(\delta_0, C_0)$, $E_1:=\re E$, we plan to construct microlocal solutions to the (matrix) equation $(P-E)u=0$, that are concentrated near some arbitrary point of  $\Gamma_1(E_1)\cup \Gamma_2(E_1)$, where, for $j=1,2$, we have set,
$$
\Gamma_j(E_1):=\{(x,\xi) \in \R^2\, ;\, \xi^2 +V_j(x)=E_1\}.
$$
Actually, since we plan to compare these solutions with the resonant states obtained from Theorem \ref{ThapproxRes},
 we will construct them in such a way that they are ``out-going'', which in this case means
 that they have no microsupport on the in-coming set $\Gamma_{2,R}^-(E_1)$, defined by,
\be
\label{Gamma2-}
\Gamma_{2,R}^-(E_1):=\{(x,\xi) \in \R^2\, ;\, \xi^2 +V_2(x)=E_1\, ,\, x>0\, , \xi <0\}.
\ee
This fact is well-known in the scalar case (see [He Sj], [Be-Ma]), 
and our system is reduced to such a scalar pseudo-differential operator on $\Gamma_{2,R}^-$ thanks to the microlocal ellipticity of $P_1$.
In addition, since such constructions are standard away from the crossing points $\rho_\pm (E_1):=(0,\pm\sqrt{E_1})$ (they are the usual WKB constructions), we will concentrate on small neighbourhoods of $\rho_\pm (E_1) $.

\begin{center}
\scalebox{0.4}{
\includegraphics{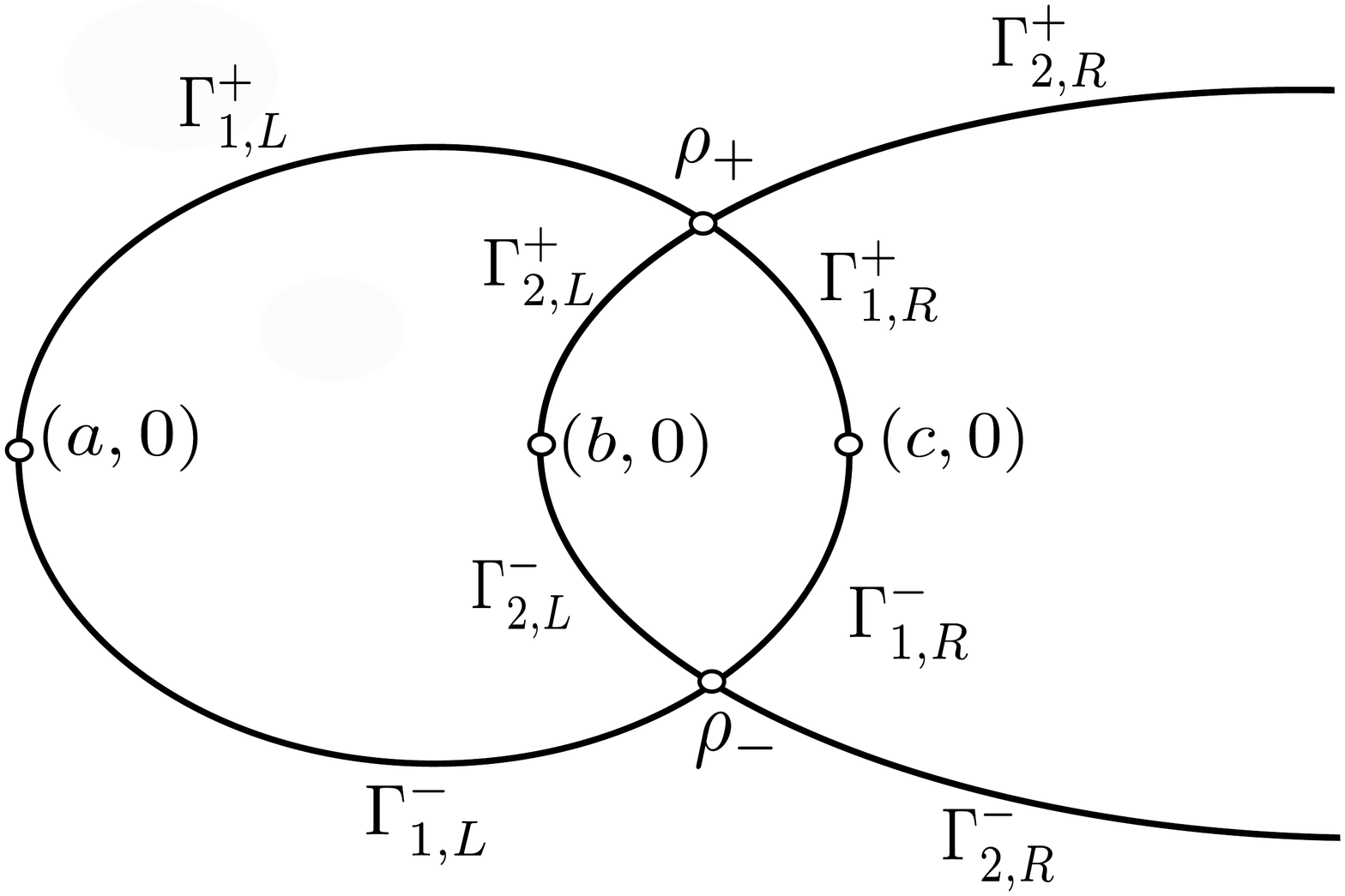}
}
\centerline{Figure 3: Characteristic sets on the phase-space}
\end{center}

At a first stage, we work with real values of $E$ only, considering $E$ as an extra-parameter. At the end, we will just observe that all our constructions depend analytically on $E$, and can be extended to complex values as long as $\im E$ remains $\cO(h)$.

For  $j=1,2$, we also set,
$$
\Gamma_{j,L}=\{(x,\xi)\in \Gamma_j(E); x<0\},\quad
\Gamma_{j,L}^\pm=\{(x,\xi)\in \Gamma_j(E); x<0,\pm\xi>0\},
$$
$$
\Gamma_{j,R}=\{(x,\xi)\in \Gamma_j(E); x>0\},\quad
\Gamma_{j,R}^\pm=\{(x,\xi)\in \Gamma_j(E); x>0,\pm\xi>0\}.
$$

We start by working near $\rho_-(E)$ where we are going to construct a basis of microlocal solutions, with one of them microlocally concentrated on $\Gamma_1(E)\cup\Gamma_2(E)\backslash \Gamma_{2,R}^-(E)$. In all of these constructions, the various pseudodifferential operators are considered at a formal level only, which means that they are identified with their formal analytic Weyl-semiclassical symbol (near a given point of $\R^{2n}$), as in \cite[Appendix]{HeSj2}. The same thing is valid for the Fourier integral operators, once their phase function is fixed. Moreover, if $A$ is such a formal analytic pseudodifferential operator, with symbol $a$ defined near $(x_0,\xi_0)\in \R^{2n}$, one can define its (microlocal) action on ${\mathcal D}'(\R)$ by re-summing its symbol up to $\cO(e^{-\alpha /h})$ with $\alpha >0$ (see, e.g., \cite{Sj}), by multiplying it by a cut-off function centered at $(x_0,\xi_0)$, and by taking its  usual Weyl-quantization. In that case, the writing
\be
\label{sim0}
Av \sim 0 \mbox{ microlocally near } (x_0,\xi_0)
\ee
just means that the microsupport $MS(Av)$ of $Av$ does not contain $(x_0,\xi_0)$ (and thus neither a neighbourhood of it). Here, the notion of microsupport is the one used, e.g., in \cite{HeSj2, Ma}. 

In particular, if $a$ vanishes near $(x_0,\xi_0)$, then \eqref{sim0} is valid for any $v\in L^2_{loc}$, possibly $h$-dependent, with its local $L^2$-norm near $x_0$ that is $\cO(h^{-N})$ for some $N>0$. In that case, we also write $A\sim 0$ microlocally near $(x_0,\xi_0)$.

Now, microlocally near $\rho_-(E)$, we investigate the solutions $v=(v_1,v_2)$ to $(P-E)v \sim 0$. Thanks to Assumption (A5), the operator $W^*$ is  elliptic at $\rho_-(E)$. Therefore, by using the symbolic calculus, we can construct a microlocal parametrix to it, that we denote by $W_*^{-1}$. Then, microlocally near $\rho_-(E)$, the system $(P-E)v \sim 0$ can be 
re-written as,
\be
\label{systemrewrite}
\left\{ \begin{aligned}
& v_1\sim -\frac1{h}W_*^{-1}(P_2-E)v_2;\\
& W^*(P_1-E)W_*^{-1}(P_2-E)v_2 -h^2W^*Wv_2\sim 0.
\end{aligned}
\right.
\ee
Since the principal symbol of $Q:= W^*(P_1-E)W_*^{-1}(P_2-E) -h^2W^*W$ is $(\xi^2+V_1(x)-E)(\xi^2+V_2(x)-E)$, it has a saddle point at $\rho_-(E)$, 
and we can apply Theorem b.1 of \cite{HeSj2} or, rather, its generalization to non-selfadjoint operators given in \cite[Proposition 6.1]{Ra} 
(see also \cite[Proposition 4.6]{Ba}). 
We obtain the existence of a formal  Fourier integral operator $U$ 
(with associated canonical transform ${\mathcal \kappa}$ sending $\rho_-(E)$ to $(0,0)$)
 and a formal analytic symbol $F=F(t,h)$ (defined near $t=0$), such that,
\be
\label{reduc}
UF(Q,h)U^{-1} \sim \frac12(yhD_y + hD_y\cdot y)=:G_0 \mbox{ microlocally near } (0,0).
\ee
(Here, the formal pseudodifferential operator $F(Q,h)$ is defined as in \cite[Appendix]{HeSj2}.) 

In particular, $\widetilde v_2:= Uv_2$ is a solution to,
\be
\label{eqModel}
G_0\widetilde v_2 \sim F(0,h)\widetilde v_2 \, \mbox{ microlocally near } (0,0).
\ee
Actually, there are several ways of reducing $Q$ to $G_0$, depending on which affine transformation is used for sending $T_{\rho_-(E)}\Gamma_1(E) \cup T_{\rho_-(E)}\Gamma_2(E)$ onto $\{ y\eta =0\}$.
We show,
\begin{proposition}\sl 
The reduction \eqref{reduc} can be made in such a way that one has, 
$$
F(0,h) = -\frac{i}2 h +\cO( h^2)
$$
uniformly with respect to $h>0$ small enough and $E\in{\mathcal D}_h(\delta_0, C_0)$. 
\end{proposition}
\begin{proof} Setting $\tau_1 =V'_1(0)$ and $\tau_2= - V_2'(0)$ 
(both positive by Assumption (A4)), we obtain that, near $\rho_-(E)$, the principal symbol $q_0$ of $Q$ satisfies,
$$
q_0(x,\xi ) = 4E(\xi+\sqrt{E}-\frac{\tau_1}{2\sqrt E}x) (\xi+\sqrt{E}+\frac{\tau_2}{2\sqrt E}x)+\cO(|x|^3 + |\xi +\sqrt{E}|^3).
$$
As it can be seen in the proof of \cite[Theoreme b.1]{HeSj2}, the first step in the reduction \eqref{reduc} consists in transforming $q_0$ into $\alpha y\eta +\cO(|(y,\eta)|^3)$ (with $\alpha \not= 0$ constant) by means of an affine canonical transformation $\kappa_0$. Here, we choose $\kappa_0(x,\xi)=(y,\eta)$ defined by,
\be
\label{choixrot}
\left\{
\begin{aligned}
&y=\frac{\tau_1}{2\sqrt E}x-\xi -\sqrt{E} ;\\
& \eta =  \frac{\tau_2}{\tau_1+\tau_2} x + \frac{2\sqrt E}{\tau_1+\tau_2}(\xi + \sqrt{E}).
\end{aligned}
\right.
\ee
One can immediately check that $\kappa_0$ is symplectic, 
and that under this change $q_0(x,\xi)$ becomes 
$-2(\tau_1+\tau_2)\sqrt{E}\, y\eta +\cO(|(y,\eta)|^3)$. 
Then, the next steps in the proof of \cite[Theoreme b.1]{HeSj2} 
mainly consist in correcting this choice by terms that are 
$\cO(|(y,\eta)|^2)$, and in constructing $F(t,h)$ in order to finally obtain \eqref{reduc}. In particular, the canonical transformation $\kappa$ associated to $U$ satisfies,
\be
\label{approxkappa}
\kappa (x,\xi )=\kappa_0(x,\xi) + \cO(|x|^2 + |\xi +\sqrt{E}|^2).
\ee
Moreover, by the Weyl-symbolic calculus (see \cite{HeSj2}), 
we also know that the Weyl-symbol $\sigma$ of $UF(Q,h)U^{-1}$ satisfies,
$$
\sigma (y,\eta;h) = F(q(\kappa^{-1}(y,\eta )),h) + \cO(h^2),
$$
where $q$ stands for the Weyl-symbol of $Q$. Thus, denoting by $\sum_{k\geq 0}h^k f_k(t)$ the semiclassical asymptotic expansion of $F(t,h)$, and $q_1$ the subprincipal symbol of $Q$, by a Taylor expansion we deduce from \eqref{reduc},
\be
\label{eqF0F1}
\begin{aligned}
 & f_0(q_0(\kappa^{-1}(y,\eta )) )=y\eta\\
 & f_1(q_0(\kappa^{-1}(y,\eta )) ) + q_1(\kappa^{-1}(y,\eta ))f_0'(q_0(\kappa^{-1}(y,\eta )) ) =0.
 \end{aligned}
\ee
In particular, we have $f_0(0)=0$, and,  applying $\partial_y\partial_\eta$ to the first equation of \eqref{eqF0F1} and taking the value at $y=\eta=0$, we also obtain (using that $q_0(\kappa^{-1}(y,\eta ))=-2(\tau_1+\tau_2)\sqrt{E}\, y\eta +\cO(|(y,\eta)|^3)$),
\be
\label{f'(0)}
-2(\tau_1+\tau_2)\sqrt{E} f_0'(0) = 1.
\ee
Then, observing that $q_1(x,\xi) = \frac1{2i}\{ \xi^2 + V_1(x), \xi^2+ V_2(x)\} + \cO(|\xi^2 + V_1(x)|+|\xi^2 + V_2(x)|)=\frac1{i}\xi (V_2'(x)-V_1'(x))+ \cO(|\xi^2 + V_1(x)|+|\xi^2 + V_2(x)|)$ (so that $q_1(\kappa^{-1}(0,0 ))=-i\sqrt{E}(\tau_1+\tau_2)$), we deduce,
$$
f_1(0)=i\sqrt{E}(\tau_1+\tau_2)f_0'(0)= -\frac{i}2,
$$
and the result follows.
\end{proof}
\begin{remark}\sl If in \eqref{choixrot} we had exchanged the roles of $y$ and $\eta$ (e.g. by the symplectic change $(y,\eta) \mapsto (\eta, -y)$), we would have obtained $F(0,h) = \frac{i}2 h+\cO(h^2)$ and the equation \eqref{eqModel} would have become $yd_y\widetilde v_2=-(1+\cO(h))\widetilde v_2$, making
 the construction of the solution less easy (as we will see, with our choice, instead, \eqref{eqModel} becomes $yd_y\widetilde v_2=\cO(h)\widetilde v_2$).
\end{remark}

In that case, we also have,
\begin{proposition}\sl 
\label{propF(0)} 
It holds that 
$$
F(0,h) = -\frac{i}2 h +\mu h^2,
$$
with 
$$
\mu =\mu(h) = -\frac{r_0(0)^2+r_1(0)^2E}{2(\tau_1+\tau_2)\sqrt{E} } + \cO(h)
$$
uniformly with respect to $h>0$ small enough and $E\in{\mathcal D}_h(\delta_0, C_0)$. 
\end{proposition}
\begin{proof} We first observe that, under the conjugation by $U$, the characteristic set of  $P_2-E$ is changed into $\{\eta=0$\}. In particular, it projects bijectively on the base $\R_y$ near $(0,0)$, and thus the microlocal solutions to $(P_2-E)u=0$ near $\rho_-(E)$ are transformed into smooth functions of $y$. 
If we set $v=h^{-1/6}u_{2,L}^-$ (so that $v(0)$ is $\cO(1)$), 
then $v$ is solution to $Qv=-h^2W^*Wv$, and  the smooth function $\widetilde v:=Uu_{2,L}^-$ is solution to,
\be
G_0 \widetilde v \sim UF(-h^2W^*W,h)U^{-1} \widetilde v \, \mbox{ microlocally near } (0,0).
\ee
Writing $F(t,h) = F(0,h) + t F_1(t,h)= F(0,h) + t(\widetilde F_1(t) +\cO(h))$, we deduce,
\be
\label{eqmodelu2}
\left(yhD_y-\frac{ih}2\right)\widetilde v\sim F(0,h)\widetilde v-h^2\widetilde F_1(0)UW^*WU^{-1}\widetilde v +h^3R\widetilde v,
\ee
microlocally near $(0,0)$, where $R$ is a $0$-th order semiclassical pseudodifferential operator. 
Applying the FBI transform $T\, :\, w \mapsto Tw$ given by,
$$
Tw(y,\eta;h):=h^{-\frac12}\int e^{i(y-y')\eta /h - (y-y')^2/2h}\chi (y')w(y')dy',
$$
(where $\chi$ is a cut-off function near $0$), integrating by parts, and taking the value at $y=\eta=0$, we deduce from \eqref{eqmodelu2},
\be
\label{F(0,h)modh3prelim}
-\frac{ih}2 T\widetilde v(0,0)=\left( F(0,h)-h^2\widetilde F_1(0)|W(\rho_-(E))|^2\right) T\widetilde v(0,0)+\cO(h^3),
\ee
where we have denoted by $W(x,\xi):=r_0(x)+ir_1(x)\xi$ the principal symbol of $W(x,hD_x)$.\\
Since the microsupport of $\widetilde v$ near $(0,0)$ coincides with $\eta =0$, and $\widetilde v(0) \sim 1$, 
a stationary-phase expansion shows that $T\widetilde v(0,0)\sim 1$, too, as $h\to 0_+$. 
Therefore \eqref{F(0,h)modh3prelim} implies,
\be
\label{F(0,h)modh3prelim2}
F(0,h) = -\frac{ih}2+h^2\widetilde F_1(0)|W(\rho_-(E))|^2+\cO(h^3).
\ee
On the other hand, by definition we have,
$$
\widetilde F_1(0)=F_1(0,h)+\cO(h) = \partial_tF(0,h)+\cO(h) = f_0'(0) +\cO(h),
$$
and the result follows from \eqref{f'(0)} and \eqref{F(0,h)modh3prelim2}.
\end{proof}

Now, thanks to Proposition \ref{propF(0)}, we can re-write \eqref{eqModel} as,
\be
\label{eqModelbis}
y\frac{ d\widetilde v_2}{dy}\sim i\mu h\, \widetilde v_2  \, \mbox{ microlocally near } (0,0).
\ee
A basis of solutions to \eqref{eqModelbis} is given by the two functions,
$$
u^\dashv (y) := H(-y)|y|^{i\mu h}\quad ; \quad u^\vdash (y):= H(y)|y|^{i\mu h},
$$
where $H$ stands for the Heaviside function. In particular, near $(0,0)$, their microsupports satisfy,
$$
\begin{aligned}
& MS(u^\dashv (y)) = \{ y< 0,\, \eta =0\}\cup \{ y=0\};\\
& MS(u^\vdash (y)) = \{ y> 0,\, \eta =0\}\cup \{ y=0\}.
\end{aligned}
$$
As a consequence, by construction the functions 
$v_2^\dashv :=U^{-1}u^\dashv $ and $v_2^\vdash :=U^{-1}u^\vdash$ 
are solutions to $Qv_2\sim0$ microlocally  in a small neighborhood 
${\mathcal V}_-$ of $\rho_-(E)$, and their microsupports satisfy,
$$
\begin{aligned}
& MS(v_2^\dashv ( x)) \cap{\mathcal V}_- \subset (\Gamma_1(E)\cup\Gamma_{2,L}^-(E)) \cap{\mathcal V}_-;\\
& MS(v_2^\vdash ( x)) \cap{\mathcal V}_- \subset (\Gamma_1(E)\cup\Gamma_{2,R}^-(E)) \cap{\mathcal V}_-.
\end{aligned}
$$
Going back to \eqref{systemrewrite}, if we set $v_1^\dashv:= -\frac1{h}W_*^{-1}(P_2-E)v_2^\dashv$, $v_1^\vdash:= -\frac1{h}W_*^{-1}(P_2-E)v_2^\vdash$, $v^\dashv:= (v_1^\dashv,v_2^\dashv)$, and $v^\vdash= (v_1^\vdash,v_2^\vdash)$, then $v^\dashv$ and $v^\vdash$ are both solutions to $(P-E)v\sim 0$ 
microlocally  in  ${\mathcal V}_-$, and their microsupports satisfy,
$$
\begin{aligned}
& MS(v^\dashv ( x)) \cap{\mathcal V}_- \subset  (\Gamma_1(E)\cup\Gamma_{2,L}^-(E))
 \cap{\mathcal V}_-;\\
& MS(v^\vdash ( x)) \cap{\mathcal V}_- \subset  (\Gamma_1(E)\cup\Gamma_{2,R}^-(E))
 \cap{\mathcal V}_-.
\end{aligned}
$$

Then we plan to compute the microlocal asymptotic behaviour of $v^\dashv$ and $v^\vdash$ on each of the three branches of their microsupport. 

We first observe that, by a convenient normalization of the Fourier integral operator $U$, we can assume,
$$
U^{-1} u (x;h) = \int_\R e^{i\psi (x,y)/h} c(x,y;h)u (y) dy,
$$ 
where $c\sim \sum_{k\geq 0}h^kc_k$ is an analytic symbol, $c_0(0,0)=1$, and $\psi$ is a generating function of $\kappa^{-1}$, in the sense that we have  $\kappa^{-1}\, :\, (y, -\nabla_y\psi) \mapsto (x, \nabla_x\psi )$.

Now we compute the asymptotic behaviour of the solutions $v^\dashv$ and $v^\vdash$ on each of $\Gamma_{1,L}^- (E)$, $\Gamma_{1,R}^- (E)$, $\Gamma_{2,R}^- (E)$, or $\Gamma_{2,L}^- (E)$. Microlocally near these curves, the operator
$Q$ is of principal type (it can be microlocally transformed into $hD_y$), and thus the space of microlocal solutions to $(P-E)u=0$ is one-dimensional. 
In addition, using the phase functions $-\nu_j^0(x)$ ($j=1,2$), it is not difficult to construct microlocal WKB solutions there, of the type $a^-(x;h)e^{-i\nu_1^0(x)/h}$ on 
$\Gamma_{1,R}^- (E)$ or $\Gamma_{1,L}^- (E)$, and of the type $b^-(x;h)e^{-i\nu_2^0(x)/h}$ on $\Gamma_{2,R}^- (E)$ or $\Gamma_{2,L}^- (E)$, where $a$ and $b$ are elliptic analytic (vector-valued) symbols that become singular at $x=0$. More precisely,
\begin{proposition}
\label{wkb-}
There exist functions $f_{1,L}^-, f_{1,R}^-, f_{2,L}^-, f_{2,R}^-\in L^2(\R)$ such that 
$$
(P-E)f_{1,L}^-\sim 0,\quad f_{1,L}^-\sim 
\left (
\begin{array}{c}
a_1^{-} \\[8pt]
ha_2^{-}
\end{array}
\right )
e^{- i\nu_1^0(x)/h}
\,\,\,{\rm microlocally\, on}\,\,
\Gamma_{1,L}^-(E),
$$
$$
(P-E)f_{1,R}^-\sim 0,\quad f_{1,R}^-\sim 
\left (
\begin{array}{c}
a_1^{-} \\[8pt]
ha_2^{-}
\end{array}
\right )
e^{- i\nu_1^0(x)/h}
\,\,\,{\rm microlocally\, on}\,\,
\Gamma_{1,R}^-(E),
$$
$$
(P-E)f_{2,L}^-\sim 0,\quad f_{2,L}^-\sim 
\left (
\begin{array}{c}
hb_1^{-} \\[8pt]
b_2^{-}
\end{array}
\right )
e^{- i\nu_2^0(x)/h}
\,\,\,{\rm microlocally\, on}\,\,
\Gamma_{2,L}^-(E),
$$
$$
(P-E)f_{2,R}^-\sim 0,\quad f_{2,R}^-\sim 
\left (
\begin{array}{c}
hb_1^{-} \\[8pt]
b_2^{-}
\end{array}
\right )
e^{- i\nu_2^0(x)/h}
\,\,\,{\rm microlocally\, on}\,\,
\Gamma_{2,R}^-(E).
$$
Here, 
$$
\nu_j^0(x):= \int_0^x \sqrt{ E-V_j(t)} \,dt \quad (j=1,2),
$$
and $a_j^-= a_j^-(x;h)\sim \sum_{k\geq 0}h^ka_{j,k}^-(x)$ and $b_j^-=b_j^-(x;h)\sim \sum_{k\geq 0}h^kb_{j,k}^-(x)$ ($j=1,2$) are analytic symbols whose first coefficients are given by \eqref{coeffaj}, \eqref{coeffbj} below.
\end{proposition}

\begin{proof}
On $\Gamma_1^-  (E)$, writing $a^-=\left(\begin{array}{c} a_1^- \\ ha_2^- \end{array}\right)$ with $a_j^-(x;h)\sim \sum_{k\geq 0}h^ka_{j,k}^-(x)$ ($j=1,2$), and solving the transport equations, we find,
\be
\label{coeffaj}
\begin{aligned}
& a_{1,0}^-=\frac1{(E-V_1)^{\frac14}}\quad ; \quad a_{2,0}^-=\frac{r_0+ir_1\sqrt{E-V_1}}{(V_1-V_2)(E-V_1)^{\frac14}}.
\end{aligned}
\ee
Similarly,  on $\Gamma_2^-(E)$, writing $b^-=\left(\begin{array}{c} hb_1^- \\ b_2^- \end{array}\right)$ with $b_j^-\sim \sum_{k\geq 0}h^kb_{j,k}^-(x)$ ($j=1,2$), we obtain,
\be
\label{coeffbj}
\begin{aligned}
& b_{2,0}^-=\frac1{(E-V_2)^{\frac14}}\quad ; \quad b_{1,0}^-=\frac{r_0-ir_1\sqrt{E-V_2}}{(V_2-V_1)(E-V_2)^{\frac14}}.
\end{aligned}
\ee
\end{proof}

\begin{proposition}\sl 
\label{exprho-}
Let ${\mathcal V}_-$ be a small enough neighbourhood of $\rho_-(E)$. 
The functions $v^\dashv$ and $v^\vdash$ have the following microlocal asymptotic behaviours:
\be
\label{vgauche}
\begin{aligned}
& v^\dashv \sim \alpha_R^\dashv h^{i\mu h}
 f_{1,R}^- \, \mbox{ microlocally on } \Gamma_{1,R}^- (E)\cap {\mathcal V}_-;\\
& v^\dashv \sim \alpha_L^\dashv h^{i\mu h}  f_{1,L}^- \, \mbox{ microlocally on } \Gamma_{1,L}^- (E)\cap {\mathcal V}_-;\\
& v^\dashv \sim \beta_L^\dashv \sqrt{h} f_{2,L}^- \, \mbox{ microlocally on } \Gamma_{2,L}^- (E)\cap {\mathcal V}_-,
\end{aligned}
\ee
and,
\be
\label{vgdroite}
\begin{aligned}
& v^\vdash \sim \alpha_R^\vdash h^{i\mu h} f_{1,R}^- \, \mbox{ microlocally on } \Gamma_{1,R}^- (E)\cap {\mathcal V}_-;\\
& v^\vdash \sim \alpha_L^\vdash h^{i\mu h} f_{1,L}^- \, \mbox{ microlocally on } \Gamma_{1,L}^- (E)\cap {\mathcal V}_-;\\
& v^\vdash \sim \beta_R^\vdash\sqrt{h} f_{2,R}^- \, \mbox{ microlocally on } \Gamma_{2,R}^- (E)\cap {\mathcal V}_-,
\end{aligned}
\ee
with $\alpha_S^\dashv= \alpha_S^\dashv(h)\sim  \sum_{k\geq 0}h^k\alpha_{S,k}^\dashv $, $\beta_S^\dashv=\beta_S^\dashv(h) \sim \sum_{k\geq 0}h^k\beta_{S,k}^\dashv $, $\alpha_S^\vdash= \alpha_S^\vdash(h)\sim  \sum_{k\geq 0}h^k\alpha_{S,k}^\vdash $, $\beta_S^\vdash=\beta_S^\vdash(h) \sim \sum_{k\geq 0}h^k\beta_{S,k}^\vdash $ ($S=L,R$),
\be
\label{alphabeta1}
\begin{aligned}
&  \alpha_{L,0}^\dashv = \alpha_{R,0}^\dashv =  - \alpha_{L,0}^\vdash=-\alpha_{R,0}^\vdash=\frac{i(\tau_1+\tau_2)E^{\frac14}}{r_0(0)+ir_1(0)\sqrt{E}};\\
& \beta_{L,0}^\dashv =\beta_{R,0}^\vdash= \sqrt{\pi (\tau_1+\tau_2)}e^{i\pi /4}.
\end{aligned}
\ee
\end{proposition}
\begin{proof}
We just have to  determine the coefficients of microlocal proportionality between $v^\dashv$ (respectively $v^\vdash$) and the WKB solutions.

Using \eqref{choixrot}-\eqref{approxkappa}, we see that the phase function $\psi$ in the definition of $U^{-1}$ satisfies,
\be
\label{exppsi}
\psi (x,y) =\frac{\tau_1}{4\sqrt E}x^2 + \frac{\sqrt E}{\tau_1+\tau_2}y^2 -xy -\sqrt{E}\, x + \cO(|(x,y)|^3).
\ee
Near 0 the map $y\mapsto \psi (x,y)$ admits a unique critical point $y_c(x)$, that is non-degenerate and is such that,
$$
y_c(x) = \frac{\tau_1+\tau_2}{2\sqrt E} x + \cO(x^2).
$$
In particular, the corresponding critical value $\psi(x,y_c(x))$ vanishes at $x=0$, 
and since we know that  $\kappa^{-1}$ sends $\{\eta =0\}$ onto $\Gamma_2(E)\cap {\mathcal V}_-$, we necessarily have 
$$\psi (x,y_c(x)) = -\nu_2^0(x).$$


Working with $v^\dashv$ on $x<0$, by definition we have,
$$
v_2^\dashv(x;h) =\int_{-\infty}^0 e^{i\psi (x,y)/h} c(x,y;h)|y|^{i\mu h} dy.
$$
Only two points contribute, up to exponentially small quantities, to this integral. One is the critical point $y_c(x)$ of $\psi$, and the other one is the singular point $y=0$  of $u^\dashv$.
By the stationary-phase theorem, the  contribution of $y_c(x)$ is of  the type
$$
\sqrt{2\pi h}\, \sigma^\dashv (x;h)e^{-i\nu_2^0(x)/h},
$$
where $\sigma^\dashv(x;h)\sim \sum_{k\geq 0} h^k \sigma^\dashv_k(x)$ is a classical symbol of order 0, with $\sigma_0^\dashv (x) = e^{i\pi /4}(\partial_y^2 \psi (x, y_c(x)))^{-\frac12}c_0(x, y_c(x))$ (here we use the fact that $\partial_y^2\psi (0,0)=\frac{2\sqrt E}{\tau_1+\tau_2} >0$). 

On the other hand, the contribution of $y=0$ can be computed by a complex change of contour of integration, 
that reduces the integral to a Laplace transform (see, e.g., \cite{Er}). After the change of scale $t\to ht$, we obtain an expression of the type,
$$
h^{1+i\mu h}\widetilde \sigma^\dashv (x;h) e^{i\psi (x,0)/h}
$$
where $\widetilde \sigma^\dashv(x;h)\sim \sum_{k\geq 0} h^k \widetilde \sigma^\dashv_k(x)$ is a classical symbol of order 0, 
with $\widetilde \sigma_0^\dashv (x) = ix^{-1}c_0(x, 0)$. Since $\kappa^{-1}$ sends $\{ y=0\}$ onto $\Gamma_1(E)\cap{\mathcal V}_-$ and $\psi (0,0)=0$, we necessarily have $\psi (x,0)=-\nu_1^0(x)$, and thus, locally on  $\{x<0\}$, we finally obtain,
$$
v_2^\dashv(x;h) = \sqrt{2\pi h}\, \sigma^\dashv (x,h)e^{-i\nu_2^0(x)/h}+h^{1+i\mu h}\widetilde \sigma^\dashv (x;h) e^{-i\nu_1^0(x)/h}+\cO(e^{-\varepsilon /h}),
$$
with $\varepsilon >0$ constant. In particular, using \eqref{systemrewrite}, this also gives the asymptotic behaviour of $v_1^\dashv$ there, and thus also the microlocal asymptotic behaviour of $v^\dashv$ on $\Gamma_{1,L}^- (E)$ and on $\Gamma_{2,L}^- (E)$. Comparing with the microlocal  WKB solutions we deduce the existence of constants  $\alpha_L^\dashv\sim \sum_{k\geq 0}h^k\alpha_{L,k}^\dashv $ and $\beta_L^\dashv\sim \sum_{k\geq 0}h^k\beta_{L,k}^\dashv $, such that \eqref{vgauche} holds, and thus in particular,
$$
\sqrt{2\pi }\, \sigma^\dashv  = \beta_L^\dashv  \,b_2^-\quad ; \quad \widetilde \sigma^\dashv  =\alpha_L^\dashv a_2^-.
$$
At the level of principal symbols, we obtain,
$$
\begin{aligned}
& \sqrt{2\pi }\,e^{i\pi /4}(\partial_y^2 \psi (x, y_c(x)))^{-\frac12}c_0(x, y_c(x))= \beta_{L,0}^\dashv (E-V_2(x))^{-\frac14};\\
& ix^{-1}c_0(x, 0)= \alpha_{L,0}^\dashv \frac{r_0(x)+ir_1(x)\sqrt{E-V_1(x)}}{(V_1(x)-V_2(x))(E-V_1(x))^{\frac14}},
\end{aligned}
$$
where, at a first stage, the equalities are valid for $x<0$ only. But since all the functions extend analytically to a whole neighbourhood of $0$, and $V_1(x)-V_2(x) = (\tau_1+\tau_2)x+\cO(x^2)$ , we can multiply the second equality by $x$ and  make $x$ tend to 0. Then,  using \eqref{exppsi} and the fact that $c_0(0,0)=1$, we obtain the values of $\alpha_{L,0}^\dashv $ and $\beta_{L,0}^\dashv $ given in \eqref{alphabeta1}. 

The same arguments hold for $x>0$, with the difference that only $y=0$ contributes to the integral (in this case the critical point does not belong to the interval of integration), and we obtain the expression of $\alpha_{R}^\dashv $ given in \eqref{alphabeta1}.

The result for $v^\vdash$ follows in the same way.
\end{proof}

Now we exchange the role of the indices 1 and 2. Namely, we reduce the system to a scalar equation with unknown $v_1$:
\be
\label{systemrewrite2}
\left\{ \begin{aligned}
& v_2\sim -\frac1{h}W^{-1}(P_1-E)v_1;\\
& \widehat Qv_1:=W(P_2-E)W^{-1}(P_1-E)v_1 -h^2WW^*v_1\sim 0.
\end{aligned}
\right.
\ee
The equation $\widehat Q v_1\sim 0$ is reduced to $G_0\widehat v_1\sim \widehat F(0,h)\widehat v_1$
with $\widehat v_1:= \widehat U v_1$, where the Fourier integral operator $\widehat U$ is given by,
$$
\widehat U^{-1}u(x,h)=\int_\R e^{i\widehat\psi(x,y)/h}\widehat c(x,y;h)u(y)dy,
$$
$$
\widehat\psi(x,y)=-\frac{\tau_2}{4\sqrt E}x^2-\frac{\sqrt E}{\tau_1+\tau_2}y^2 + xy - \sqrt Ex+\cO(|(x,y)|^3),
$$
$$
\widehat c\sim  \sum_{k=0}^\infty h^kc_k(x,y),\quad \widehat c_0(0,0)=1.
$$
and the formal analytic symbol $\widehat F(t,h)$ is such that,
$$
{\widehat F(0,h)=-\frac{i}{2}h + \widehat \mu h^2,}
$$
with
$$
\widehat \mu= +\frac {r_0(0)^2+r_1(0)^2E}{2(\tau_1+\tau_2)\sqrt E}+\cO(h).
$$
Notice that $\widehat \mu = -\mu$ modulo $\cO(h)$. 
Let us define as before
$$\widehat u^\dashv(y):=H(-y)|y|^{i\widehat\mu h}, \quad \widehat u^\vdash(y):=H(y)|y|^{i\widehat\mu h},
$$
$$
\widehat v_1^\dashv:=\widehat U^{-1}\widehat u^\dashv,\quad 
\widehat v_1^\vdash:=\widehat U^{-1}\widehat u^\vdash,
$$
$$
\widehat v_2^\dashv=-\frac1{h}W^{-1}(P_1-E)\widehat v_1^\dashv,\quad 
\widehat v_2^\vdash=-\frac1{h}W^{-1}(P_1-E)\widehat v_1^\vdash.
$$
Then $\widehat v^\dashv:=(\widehat v_1^\dashv,\widehat v_2^\dashv)$ and 
$\widehat v^\vdash:=(\widehat v_1^\vdash,\widehat v_2^\vdash)$ are both solutions to $(P-E)v\sim 0$ microlocally in a small neighborhood of $\rho_-(E)$ that we still denote by  ${\mathcal V}_-$,
and their microsupports satisfy,
$$MS(\widehat v^\dashv)\cap {\mathcal V}_- \subset (\Gamma_2(E)\cup\Gamma_{1,L}^-(E))\cap{\mathcal V}_-;$$
$$MS(\widehat v^\vdash)\cap {\mathcal V}_- \subset (\Gamma_2(E)\cup\Gamma_{1,R}^-(E))\cap{\mathcal V}_-.$$

\begin{proposition}\sl 
\label{exprho-bis}
The functions $\widehat v^\dashv$ and $\widehat v^\vdash$ have the following microlocal asymptotic behaviours:
\be
\label{vgauche2}
\begin{aligned}
& \widehat v^\dashv \sim \widehat \alpha_R^\dashv h^{i \hat{\mu} h}
f_{2,R}^- \, \mbox{ microlocally on } \Gamma_{2,R}^- (E)\cap {\mathcal V}_-;\\
& \widehat v^\dashv \sim \widehat \alpha_L^\dashv h^{i\hat{\mu}  h} f_{2,L}^- \, \mbox{ microlocally on } \Gamma_{2,L}^- (E)\cap {\mathcal V}_-;\\
& \widehat v^\dashv \sim \widehat \beta_L^\dashv \sqrt{h} f_{1,L}^- \, \mbox{ microlocally on } \Gamma_{1,L}^- (E)\cap {\mathcal V}_-,
\end{aligned}
\ee
and,
\be
\label{vgdroite2}
\begin{aligned}
& \widehat v^\vdash \sim \widehat \alpha_R^\vdash h^{i \hat{\mu} h} f_{2,R}^- \, \mbox{ microlocally on } \Gamma_{2,R}^- (E)\cap {\mathcal V}_-;\\
& \widehat v^\vdash \sim \widehat \alpha_L^\vdash h^{i \hat{\mu} h} f_{2,L}^- \, \mbox{ microlocally on } \Gamma_{2,L}^- (E)\cap {\mathcal V}_-;\\
& \widehat v^\vdash \sim \widehat \beta_R^\vdash\sqrt{h}f_{1,R}^- \, \mbox{ microlocally on } \Gamma_{1,R}^- (E)\cap {\mathcal V}_-,
\end{aligned}
\ee
with $\widehat \alpha_S^\dashv= \widehat \alpha_S^\dashv(h)\sim  \sum_{k\geq 0}h^k\widehat \alpha_{S,k}^\dashv $, $\widehat \beta_S^\dashv=\widehat \beta_S^\dashv(h) \sim \sum_{k\geq 0}h^k\widehat \beta_{S,k}^\dashv $, $\widehat \alpha_S^\vdash= \widehat \alpha_S^\vdash(h)\sim  \sum_{k\geq 0}h^k\widehat \alpha_{S,k}^\vdash $, $\widehat \beta_S^\vdash=\widehat \beta_S^\vdash(h) \sim \sum_{k\geq 0}h^k\widehat \beta_{S,k}^\vdash $ ($S=L,R$),
\be
\label{alphabeta2}
\begin{aligned}
&  +  \widehat \alpha_{L,0}^\dashv = + \widehat \alpha_{R,0}^\dashv= - \widehat \alpha_{L,0}^\vdash= - \widehat \alpha_{R,0}^\vdash 
=  \frac{i(\tau_1+\tau_2)E^{\frac14}}{r_0(0) - ir_1(0)\sqrt{E}};\\
& \widehat \beta_{L,0}^\dashv =\widehat \beta_{R,0}^\vdash= \sqrt{\pi (\tau_1+\tau_2)}e^{-i\pi /4}.
\end{aligned}
\ee
\end{proposition}

From Proposition \ref{exprho-} and Proposition \ref{exprho-bis}, we can compute the following transfer matrix, which connects microlocal data on
$\Gamma_{1,L}^-$ and $\Gamma_{2,R}^-$ to those on $\Gamma_{1,R}^-$ and $\Gamma_{2,L}^-$.

\begin{proposition}
\label{connection1}
Suppose $u(x,h)\in L^2(\R)$ satisfies $(P-E)u\sim 0$ microlocally in a neighbourhood ${\mathcal V}_-$ of $\rho_-(E)$ and set
$$
u\sim t_{1,L}^{-}f_{1,L}^{-}
\,\,\,{\rm microlocally\, on}\,\,
\Gamma_{1,L}^{-}(E),
$$
$$
u\sim t_{1,R}^{-}f_{1,R}^{-}
\,\,\,{\rm microlocally\, on}\,\,
\Gamma_{1,R}^{-}(E),
$$
$$
u\sim t_{2,L}^{-}f_{2,L}^{-}
\,\,\,{\rm microlocally\, on}\,\,
\Gamma_{2,L}^{-}(E),
$$
$$
u\sim t_{2,R}^{-}f_{2,R}^{-}
\,\,\,{\rm microlocally\, on}\,\,
\Gamma_{2,R}^{-}(E),
$$
for constants $t_{j,S}^-=t_{j,S}^-(E,h)$.
Then it holds that
\begin{equation}
\label{tau-}
\left (
\begin{array}{c}
t_{1,R}^{-} \\[8pt]
t_{2,L}^{-}
\end{array}
\right )
=
\left (
\begin{array}{cc}
\tau_{1,1}^-(E,h) & \tau_{1,2}^-(E,h) \\[8pt]
\tau_{2,1}^-(E,h) & \tau_{2,2}^-(E,h)
\end{array}
\right )
\left (
\begin{array}{c}
t_{1,L}^{-} \\[8pt]
t_{2,R}^{-}
\end{array}
\right ),
\end{equation}
where the coefficients $\tau_{j,k}^-$ have the asymptotic behaviour
$$
\begin{array}{ll}
\tau^-_{1,1}=1+\ord (h),\quad &\tau^-_{1,2}=  e^{\frac{\pi}{4} i}\tau_- h^{\frac 12+i\mu h}+\ord (h^{\frac 32}), \\[8pt]
\tau^-_{2,1}=  e^{-\frac{\pi}{4} i}\tau_+ h^{\frac 12-i\mu h}+\ord (h^\frac 32),\quad &\tau^-_{2,2}=1+\ord (h),
\end{array}
$$
as $h\to 0_+$ uniformly for $E\in{\mathcal D}_h(\delta_0, C_0)$. Here $\tau_\pm$ are constants defined by 
$$
 \tau_\pm=\sqrt\frac\pi{\tau_1+\tau_2}\left (r_0(0)E^{-\frac 14} \pm ir_1(0)E^{\frac 14}\right ).
$$
\end{proposition}
\begin{proof}
If $t_{1,L}^-=1$ and $t_{2,R}^-=0$, then $u$ should be equal to $(\alpha_L^{\dashv}h^{i\mu h})^{-1}v^{\dashv}$ microlocally in a neighbourhood of $\rho^-$,
and hence we have
$$
\tau_{1,1}^-=t_{1,R}^-=\frac{\alpha_R^\dashv}{\alpha_L^\dashv},\quad \tau_{2,1}^-=t_{2,L}^-=\frac{\beta_L^\dashv}{\alpha_L^\dashv}h^{\frac 12-i\mu h}.
$$
Similarly, if $t_{1,L}^-=0$ and $t_{2,R}^-=1$, then $u$ should be equal to $(\widehat \alpha_R^{\vdash}h^{i\widehat\mu h})^{-1}\widehat v^{\vdash}$ microlocally in a neighbourhood of $\rho^-$,
and hence we have
$$
\tau_{1,2}^-=t_{1,R}^-=\frac{\widehat\beta_R^\vdash}{\widehat\alpha_R^\vdash}h^{\frac 12-i\widehat\mu h},\quad \tau_{2,2}^-=t_{2,L}^-=\frac{\widehat\alpha_L^\vdash}{\widehat\alpha_R^\vdash}.
$$
Then Proposition \ref{connection1} follows from \eq{alphabeta1} and \eq{alphabeta2}.
\end{proof}

We obtain in the same way  the transfer matrix near $\rho_+(E)$. First, we define WKB solutions microlocally on the characteristics
$\Gamma_{j,S}^+$ as Proposition \ref{wkb-}:
\begin{proposition}
\label{wkb+}
There exist functions $f_{1,L}^+, f_{1,R}^+, f_{2,L}^+, f_{2,R}^+\in L^2(\R)$ such that 
$$
(P-E)f_{1,L}^+\sim 0,\quad f_{1,L}^+\sim 
\left (
\begin{array}{c}
a_1^{+} \\[8pt]
ha_2^{+}
\end{array}
\right )
e^{i\nu_1^0(x)/h}
\,\,\,{\rm microlocally\, on}\,\,
\Gamma_{1,L}^+(E),
$$
$$
(P-E)f_{1,R}^+\sim 0,\quad f_{1,R}^+\sim 
\left (
\begin{array}{c}
a_1^{+} \\[8pt]
ha_2^{+}
\end{array}
\right )
e^{ i\nu_1^0(x)/h}
\,\,\,{\rm microlocally\, on}\,\,
\Gamma_{1,R}^+(E),
$$
$$
(P-E)f_{2,L}^+\sim 0,\quad f_{2,L}^+\sim 
\left (
\begin{array}{c}
hb_1^{+} \\[8pt]
b_2^{+}
\end{array}
\right )
e^{i\nu_2^0(x)/h}
\,\,\,{\rm microlocally\, on}\,\,
\Gamma_{2,L}^+(E),
$$
$$
(P-E)f_{2,R}^+\sim 0,\quad f_{2,R}^+\sim 
\left (
\begin{array}{c}
hb_1^{+} \\[8pt]
b_2^{+}
\end{array}
\right )
e^{ i\nu_2^0(x)/h}
\,\,\,{\rm microlocally\, on}\,\,
\Gamma_{2,R}^+(E).
$$
Here, 
$a_j^+= a_j^-(x;h)\sim \sum_{k\geq 0}h^ka_{j,k}^+(x)$ and $b_j^+=b_j^+(x;h)\sim \sum_{k\geq 0}h^kb_{j,k}^+(x)$ ($j=1,2$) are analytic symbols whose first coefficients are given by 
\be
\label{coeffaj2}
\begin{aligned}
& a_{1,0}^+=\frac1{(E-V_1)^{\frac14}}\quad ; \quad a_{2,0}^+=\frac{r_0 - ir_1\sqrt{E-V_1}}{(V_1-V_2)(E-V_1)^{\frac14}}.
\end{aligned}
\ee
\be
\label{coeffbj2}
\begin{aligned}
& b_{2,0}^+=\frac1{(E-V_2)^{\frac14}}\quad ; \quad b_{1,0}^+=\frac{r_0 + ir_1\sqrt{E-V_2}}{(V_2-V_1)(E-V_2)^{\frac14}}.
\end{aligned}
\ee
\end{proposition}

\begin{proposition}
\label{connection2}
Suppose $u(x,h)\in L^2(\R)$ satisfies $(P-E)u\sim 0$ microlocally in a neighbourhood ${\mathcal V}_+$ of $\rho_+(E)$ and set
$$
u\sim t_{1,L}^{+}f_{1,L}^{+}
\,\,\,{\rm microlocally\, on}\,\,
\Gamma_{1,L}^{+}(E),
$$
$$
u\sim t_{1,R}^{+}f_{1,R}^{+}
\,\,\,{\rm microlocally\, on}\,\,
\Gamma_{1,R}^{+}(E),
$$
$$
u\sim t_{2,L}^{+}f_{2,L}^{+}
\,\,\,{\rm microlocally\, on}\,\,
\Gamma_{2,L}^{+}(E),
$$
$$
u\sim t_{2,R}^{+}f_{2,R}^{+}
\,\,\,{\rm microlocally\, on}\,\,
\Gamma_{2,R}^{+}(E).
$$
Then  it holds that
\begin{equation}
\label{tau+}
\left (
\begin{array}{c}
t_{1,L}^{+} \\[8pt]
t_{2,R}^{+}
\end{array}
\right )
=
\left (
\begin{array}{cc}
\tau_{1,1}^+(E,h) & \tau_{1,2}^+(E,h) \\[8pt]
\tau_{2,1}^+(E,h) & \tau_{2,2}^+(E,h)
\end{array}
\right )
\left (
\begin{array}{c}
t_{1,R}^{+} \\[8pt]
t_{2,L}^{+}
\end{array}
\right ),
\end{equation}
where 
\begin{align*}
&\tau^+_{1,1}= 1+\ord (h), &&\tau^+_{1,2}= -e^{-\frac{\pi}{4}i}\tau_+ h^{\frac 12-i\mu h}+\ord (h^\frac 32),\\
&\tau^+_{2,1}= -e^{\frac{\pi}{4}i}\tau_- h^{\frac 12+i\mu h}+\ord (h^{\frac 32}), && \tau^+_{2,2}= 1+\ord (h),
\end{align*}
as $h\to 0_+$ uniformly for $E\in{\mathcal D}_h(\delta_0, C_0)$. 
\end{proposition}

\section{Microlocal connection formulas and monodromy condition}

In the previous section, we defined WKB solutions $f_{j,S}^\pm(x,h)$ 
microlocally along the 8 curves $\Gamma_{j,S}^\pm$ 
on $\Gamma_1(E)\cup\Gamma_2(E)$ divided by
the 2 crossing points $\rho_\pm(E)$ and 3 turning points (caustics) $(a(E),0), (b(E),0)$ and $(c(E),0)$.
In order to know the global behaviour of solutions to the system $(P-E)u=0$, it is necessary to know the connection formulae between these WKB solutions.
The connection formulae near the crossing points were given in Proposition \ref{connection1} and \ref{connection2}, and it remains to have those between these
microlocal WKB solutions at the turning points. The following lemma is essentially due to Maslov.

\begin{lem}
\label{maslov}
Let $u$ satisfy $(P-E)u\sim 0$ microlocally near $\Gamma_{1,S}$, and suppose
$$
u\sim t_{1,S}^{+}f_{1,S}^{+}
\,\,\,{\rm microlocally\, on}\,\,
\Gamma_{1,S}^{+}(E),
$$
$$
u\sim t_{1,S}^{-}f_{1,S}^{-}
\,\,\,{\rm microlocally\, on}\,\,
\Gamma_{1,S}^{-}(E),
$$
for $S=L,R$.
Then it holds that
\begin{equation}
\label{sigmaS}
t_{1,S}^{+}=\sigma_{1,S}t_{1,S}^{-},
\end{equation}
with constants $\sigma_{1,S}$, which behave, as $h\to 0_+$,
$$
\sigma_{1,L}=-ie^{ + 2iS_{1,L}/h}+\ord (h),\quad \sigma_{1,R}=ie^{ - 2iS_{1,R}/h}+\ord (h).
$$
Here $S_{1,S}(E)$ are the action integrals defined by
$$
S_{1,L}(E) := \int_{a(E)}^0\sqrt{E-V_1(x)}dx,\quad S_{1,R}(E) := \int_0^{c(E)}\sqrt{E-V_1(x)}dx.
$$
Similarly, let $u$ be a solution to $(P-E)u\sim 0$ microlocally near 
$\Gamma_{2,L}$, and suppose
$$
u\sim t_{2,L}^{+}f_{2,L}^{+}
\,\,\,{\rm microlocally\, on}\,\,
\Gamma_{2,L}^{+}(E),
$$
$$
u\sim t_{2,L}^{-}f_{2,L}^{-}
\,\,\,{\rm microlocally\, on}\,\,
\Gamma_{2,L}^{-}(E).
$$
Then it holds that
\begin{equation}
\label{sigmaL}
t_{2,L}^{+}=\sigma_{2,L}t_{2,L}^{-},
\end{equation}
with
$$
\sigma_{2,L}=-ie^{ + 2iS_{2,L}/h}+\ord (h),
$$
where
$$
S_{2,L}(E) := \int_{b(E)}^0\sqrt{E-V_2(x)}dx.
$$
\end{lem}

\begin{proof}
We compute $\sigma_{1,L}$ here. The computations of $\sigma_{1,R}$ and $\sigma_{2,L}$ are similar.

Microlocally near $(x,\xi)=(a(E),0)$, the operator $P_2-E$ is elliptic and  invertible.
Hence the system \eq{sch} is reduced to a single semiclassical pseudo-differential equation for $u_1$:
$$
(P_1-E)u_1 - h^2Ru_1=0,
$$
where $R=W(P_2-E)^{-1}W^*$. 
The idea of Maslov is to represent a solution $u_1$ as Fourier inverse transform of a WKB function 
$e^{ig(\xi)/h}c(\xi,h)$, $c(\xi,h)\sim\sum c_k(\xi)h^k$ in the momentum variable, that is, 
$$
u_1(x,h)=\frac 1{\sqrt{2\pi h}}\int_{\R} e^{i(\tilde x\xi+g(\xi))/h}c(\xi,h)d\xi,
$$
where $\tilde x=x-a(E)$.
The WKB function 
$e^{ig(\xi)/h}c(\xi,h)$ is a solution to the semiclassical pseudo-differential equation 
\begin{equation}\label{pDEq}
Q(\xi, hD_\xi; h) (e^{ig(\xi)/h}c(\xi,h))= 0,
\end{equation} 
where $Q(\xi, hD_\xi; h)$ is the standard quantization of $q(\xi, \xi^*;h)$ given by
$$
q(\xi,\xi^*;h) = p_1(-\xi^* + a(E), \xi) - E - h^2r(-\xi^* + a(E), \xi;h),
$$
where $\xi^*$ is the dual variable of $\xi$ and $r(x,\xi;h)$ is the standard symbol of $R$. 
In order to have the phase function and the asymptotic expansion of the symbol, we write 
\begin{align*}
 & e^{-ig(\xi)/h} Q(\xi, hD_\xi; h)\left( e^{ig(\xi)/h}c(\xi,h) \right)\\
 =& \frac{1}{2\pi h} \int\!\!\int e^{i\{(\xi - \eta)\xi^* -(g(\xi) - g(\eta))\}/h} q(\xi, \xi^*; h) c(\eta,h) d\eta d\xi^*\\
 =& \frac{1}{2\pi h} \int\!\!\int e^{i(\xi-\eta)\xi^*/h} q(\xi, \xi^* + \tilde g(\xi,\eta); h) c(\eta,h) d\eta d\xi^*\\
 =& \frac{1}{2\pi h} \int\!\!\int e^{i(\xi-\eta)\xi^*/h} q(\xi, \xi^* + \tilde g(\xi,\eta); h) c_0(\eta) d\eta d\xi^*+\ord(h)
\end{align*}
where $\tilde g(\xi,\eta)$ is defined by $g(\xi) - g(\eta) = (\xi-\eta) \tilde g(\xi,\eta)$. 
Since $\tilde g(\xi,\xi) = g'(\xi)$, $\partial_\eta\tilde g(\xi,\xi)=\frac 12 g''(\xi)$  and $q(\xi, \xi^* + \tilde g(\xi,\eta); h)=\xi^2+V_1(a(E)-\tilde g(\xi,\eta))-E-\xi^*V_1'(a(E)-\tilde g(\xi,\eta))+\ord ((\xi^*)^2)+\ord(h^2)$,
the last integral in the previous identities is equal to
$$
(\xi^2+V_1(a(E)-g'(\xi))-E)c_0(\xi)
+\frac hi \partial_\eta\left(V_1'(a(E)-\tilde g(\xi,\eta))c_0(\eta)
\right)|_{\eta=\xi}+\ord(h^2).
$$
Therefore we have
\begin{align*}
 & e^{-ig(\xi)/h} Q(\xi, hD_\xi; h)\left( e^{ig(\xi)/h}c(\xi,h) \right)\\
 =& (\xi^2+V_1(a(E)-g'(\xi))-E)c_0(\xi)\\
 &\  + \frac hi \left (V_1'(a(E)- g'(\xi)) c_0'(\xi) - \frac{1}{2} g''(\xi) V_1''(a(E)-g'(\xi)) c_0(\xi)\right )+\ord(h^2).
\end{align*}

Hence the phase function $g(\xi)$ and the coefficient $c_0(\xi)$ of the symbol should satisfy respectively
the eikonal equation
\begin{equation}\label{eikonal_xi}
\xi^2+V_1(a(E)-g'(\xi))-E=0,
\end{equation}
and the first transport equation
\begin{equation}\label{1st_transport_xi}
V_1'(a(E)- g'(\xi)) c_0'(\xi) - \frac{1}{2} g''(\xi) V_1''(a(E)-g'(\xi)) c_0(\xi) = 0.
\end{equation}

Recall that, near $x=a(E)$, $V_1(x)$ behaves like
$$
V_1(x)=V_1(a(E)+\tilde x)=E-\tau_0(E)\tilde x+\ord (\tilde x^2),
$$
where $\tau_0(E)$ is an analytic function near $E=E_0$ with $\tau_0(E_0)>0$.
This means that the solution $g(\xi)$ of the eikonal equation \eqref{eikonal_xi} with $g(0)=0$ satisfies
$$
g(\xi)=-\frac 1{3\tau_0(E)}\xi^3+\ord (\xi^5),
$$
and that there are two real critical points 
of the phase $\tilde x\xi+g(\xi)$  for positive $\tilde x$, 
which are denoted by $\xi_\pm(\tilde x)$ satisfying $\xi_\pm(\tilde x) = \pm\sqrt{\tau_0(E)\tilde x} + \ord (\tilde x^{3/2})$.
The critical values $\tilde x\xi_\pm(\tilde x)+g(\xi_\pm(\tilde x))$ behave like
$$
\tilde x\xi_\pm(\tilde x)+g(\xi_\pm(\tilde x))=\pm\frac 23 \sqrt{\tau_0(E)}\tilde x^{\frac 32}+\ord(\tilde x^{\frac 52}),
$$
as $\tilde x\to 0$. Notice that $\mp g''(\xi_\pm(\tilde x))\sim \pm 2\xi_\pm(\tilde x) /{\tau_0(E_0)} >0$ for small $\tilde x>0$.

The stationary phase method gives the asymptotic expansion of $u_1$ as $h\to 0_+$ near $\tilde x$ positive and close to $0$.
Since the phase $\nu_1^0(x)$ of the WKB solution $f_{1,L}^+$ also behaves like 
$$-S_{1,L}(E) + \frac{2}{3} \sqrt{\tau_0(E)}\tilde x^{\frac{3}{2}}+\ord(\tilde x^{\frac{5}{2}})$$
as $\tilde x\to 0$,
it should coincide with $\tilde x\xi_+(\tilde x)+g(\xi_+(\tilde x))$ except the phase shift $-S_{1,L}(E)$.
Hence the contribution of the stationary phase from the critical point $\xi_+(\tilde x)$ 
$$
e^{\frac{i}{h}(\tilde x\xi_+(\tilde x)+g(\xi_+(\tilde x)) )}\sum_{k=0}^\infty \tilde c_k^+(x)h^k, 
\quad \tilde c_0^+(x)=e^{-\pi i/4}\frac{c_0(\xi_+(\tilde x))}{\sqrt{|g''(\xi_+(\tilde x))|}}
$$
should coincide with the WKB expansion $t_{1,L}^+a_1^+e^{i\nu_1^0(x)/h}$
microlocally on $\Gamma_{1,L}^+(E)$, so that one sees $e^{\frac{i}{h} S_{1,L}} \tilde c_{0}^+ = t_{1,L}^+ a_1^+$. 
Similarly the contribution from the critical point $\xi_-(\tilde x)$ 
$$
e^{\frac{i}{h}(\tilde x\xi_-(\tilde x)+g(\xi_-(\tilde x)) )}\sum_{k=0}^\infty \tilde c_k^-(x)h^k, 
\quad \tilde c_0^-(x)=e^{\pi i/4}\frac{c_0(\xi_-(\tilde x))}{\sqrt{g''(\xi_-(\tilde x))}}
$$
should coincide with the WKB expansion $t_{1,L}^-a_1^-e^{-i\nu_1^0(x)/h}$
microlocally on $\Gamma_{1,L}^-(E)$, so that one has $e^{-\frac{i}{h} S_{1,L}} \tilde c_{0}^- = t_{1,L}^- a_1^-$. 
To compute $ c_0^\pm(x)$, 
we solve the first transport equation \eqref{1st_transport_xi}. 
In fact, putting $\phi(\xi):= V_1'(a-g'(\xi))$, we see that \eqref{1st_transport_xi} is equivalent to $\left(\sqrt{\phi(\xi)}c_0(\xi) \right)' = 0$, 
so that we get explicitly 
$c_0(\xi) = \bigl(-V_1'(a - g'(\xi) ) \bigr)^{-\frac{1}{2}}$.
Hence the following holds for small $\tilde x >0$  
$$
c_0(\xi_\pm(\tilde x)) = \left(-V_1'(a + \frac{1}{\tau_0(E)} \xi_\pm(\tilde x)^2 ) \right)^{-\frac{1}{2}} (1+ \ord(\tilde x^{\frac{1}{2}})). 
$$
The above calculations and the fact that $a_{1,0}^\pm$ are independent of the sign $\pm$ 
imply that $\tilde c_0^+(x)/\tilde c_0^-(x) = -i + \ord(\tilde x^{\frac{1}{2}})$. 
Therefore we have verified $\sigma_{1,L}=-ie^{ + 2iS_{1,L}/h}+\ord (h)$. 
\end{proof}

Let $E$ be a resonance in ${\mathcal D}_h(\delta_0,C_0)$ and $u(x,h)$ a resonant state corresponding to $E$.
We normalize $u$ in such a way that 
\be
\label{norm}
u\sim f_{1,L}^- \,\,\,{\rm microlocally\,on}\,\,\Gamma_{1,L}^{-}(E).
\ee
Since $u$ is a resonant state, we also have
\be
u\sim 0 \,\,\,{\rm microlocally\,on}\,\,\Gamma_{2,R}^{-}(E).
\ee
Then using the microlocal connection formulae of the previous and this sections, we obtain
\begin{proposition}\sl
\label{outgoing}
Let $u$ be a resonant state corresponding to a resonance $E$ in ${\mathcal D}_h(\delta_0, C_0)$ satisfying
\eq{norm}. Then it holds that
\be 
u\sim t_{2,R}^+f_{2,R}^+ \,\,\,{\rm microlocally\,on}\,\,\Gamma_{2,R}^{+}(E),
\ee
with
\be\label{outgoing_coeff}
\begin{aligned}
t_{2,R}^+ &= -2i \sqrt{\frac{\pi h}{\tau_1+\tau_2}} e^{\frac{i}{h}(-S_{1,R}+S_{2,L})}\\
 & \times \!\!\left[ r_0(0) E^{-\frac{1}{4}}\! \sin\!\left(\!\frac{{\mathcal B}(E)}{h}\! +\! \frac{\pi}{4}\! \right)  
\! +\! r_1(0) E^{\frac{1}{4}}\! \cos\!\left(\!\frac{{\mathcal B}(E)}{h}\! +\! \frac{\pi}{4}\! \right) \right]
\!\! \left(\! 1\! +\! \ord(h\log \frac{1}{h})\! \right),
\end{aligned}
\ee
where 
$${\mathcal B}(E):=\int_{b(E)}^0\!\!\!\!\sqrt{E-V_2(x)}dx+\int_0^{c(E)}\!\!\!\!\!\sqrt{E-V_1(x)}dx.$$
In particular, the leading term of $| t_{2,R}^+|^2$ is equal to
\be\label{outgoing_coeff_modu}
\frac{4\pi h}{\tau_1+\tau_2} \left| r_0(0) E^{-\frac{1}{4}} \sin\left( \frac{{\mathcal B}(E)}{h} + \frac{\pi}{4}  \right)  
 + r_1(0) E^{\frac{1}{4}} \cos\left( \frac{{\mathcal B}(E)}{h} + \frac{\pi}{4} \right) \right|^2.
\ee
\end{proposition}
\begin{proof}
Using the notations in Propositions \ref{connection1}, \ref{connection2} and Lemma \ref{maslov}, we have
\begin{align*}
t_{2,R}^+&=\sigma_{1,R} \tau_{2,1}^+ \tau_{1,1}^-+\sigma_{2,L}\tau_{2,2}^+\tau_{2,1}^-, \\
 &= - \sqrt{\frac{\pi h}{\tau_1+\tau_2}}
\left(
e^{-i\theta}(a-ib) + e^{i\theta}(a+ib)
\right),
\end{align*}
where $\theta = \frac{1}{h}(S_{1,R}+S_{2,L}) - \frac{\pi}{4} - \mu h\log h$, 
$a= r_0(0) E^{-1/4}$ and $b= r_1(0) E^{1/4}$. A simple computation gives 
\begin{align*}
t_{2,R}^+ &= -2 i \sqrt{\frac{\pi h}{\tau_1+\tau_2}} e^{\frac{i}{h}(-S_{1,R}+S_{2,L})}
\left( a\cos\theta - b\sin\theta \right),\\
&= -2 i \sqrt{\frac{\pi h}{\tau_1+\tau_2}} e^{\frac{i}{h}(-S_{1,R}+S_{2,L})}
\left( a\sin(\theta+\frac{\pi}{2}) + b\cos(\theta+\frac{\pi}{2}) \right).
\end{align*}
Hence we obtain \eqref{outgoing_coeff} and \eqref{outgoing_coeff_modu}.

\end{proof}

Similarly, after having made a whole loop on $\Gamma_1(E)$, we also obtain,
\be
t_{1,L}^-=-e^{2i{\mathcal A}(E)/h}+\ord(h),
\ee
and therefore, we have,
\begin{proposition} \sl
\label{bohr}
Any resonance $E$ in ${\mathcal D}_h(\delta_0, C_0)$ satisfies
\be\label{bohr_coeff}
e^{2i{\mathcal A}(E)/h} = -1+\ord(h).
\ee
\end{proposition}


\section{Precise asymptotics of resonances}
We have already given in Theorem \ref{ThapproxRes} the semiclassical distribution of resonances modulo $\cO(h^{\frac 76})$. In this section, we give a more precise 
asymptotic behaviour of the distribution of resonances using the 
microlocal results of Propositions \ref{outgoing} and \ref{bohr}.

First, the formula \eq{reEk} about the real part of resonances is a direct consequence of Proposition \ref{bohr}.
 
Second, the formula \eq{imEk} about the imaginary part of resonances is deduced from
Proposition \ref{outgoing} and the following propositions \ref{propgreen} and \ref{propnorm}.

\begin{proposition}\sl
\label{propgreen}
Let $E$ and $u$ be as in Proposition \ref{outgoing}. Then one has, for any $x_0>c$,  
$$
-\im E=\frac{|t_{2,R}^+|^2}{\Vert u\Vert_{L^2(-\infty,x_0)}^2}h(1+\cO(h)).
$$
\end{proposition}

\begin{proof}
Let $x_0$ be any real number. The width of a resonance $E$ can  be expressed in terms of the resonant state $u(x,h)={}^t(u_1,u_2)$ and its derivative at
the point $x_0$ by the formula,
\be
\label{green}
\begin{aligned}
(\im E) \Vert u\Vert_{L^2(-\infty,x_0)}^2=& -h^2\im \left(
u_1'(x_0)\overline {u_1(x_0)}+u_2'(x_0)\overline {u_2(x_0)}\right)\\
& +h^2\im r_1(x_0)u_2(x_0)\overline {u_1(x_0)}.
\end{aligned}
\ee
Indeed this formula can be deduced easily by integration by parts starting from the identity,
$$
\im \langle (P-E)u,u\rangle_{L^2(-\infty,x_0)}=0.
$$

Suppose now that  $x_0>c$.
Suppose that $u$ is a resonant state corresponding to a resonance $E$ in ${\mathcal D}_h(\delta_0, C_0)$ satisfying
\eq{norm}. Then the asymptotic behaviour of the RHS of \eq{green} can be given in terms of the quantity $t_{2,R}^+$
of Proposition \ref{outgoing}.

Let $\chi(x,\xi)\in C_0^\infty(\R^n)$ be a cutoff function identically one in a neighborhood of the point $(x_0,\xi_0)\in \Gamma_{2,R}^+(E)$ 
and supported in its neighborhood.
Then  the function $u-t_{2,R}^+\chi^Wf_{2,R}^+$ does not have microsupport near the set
$\{x_0\}\times\R_\xi$, and it follows that
$$
u=t_{2,R}^+\chi^Wf_{2,R}^++\cO(h^\infty)
$$ 
locally in a neighbourhood of $x=x_0$.
Hence from Proposition \ref{wkb+}, one has
$$
u_1'(x_0)\overline{u_1(x_0)}-r_1(x_0)u_2(x_0)\overline{u_1(x_0)}=|t_{2,R}^+|^2\cO(h),
$$
$$
u_2'(x_0)\overline{u_2(x_0)}=|t_{2,R}^+|^2(ih^{-1}(\nu_2^0)'|b_2^+|^2+\cO(1))=ih^{-1}|t_{2,R}^+|^2(1+\cO(h)).
$$
\end{proof}

Concerning the $L^2$-norm of $u$ appearing in the previous proposition, we prove,
\begin{proposition}\sl 
\label{propnorm}
For any $x_0>c$ and for any $E\in [E_0 -\delta_0, E_0 + \delta_0]$, one has,
$$
\Vert u\Vert_{L^2(-\infty,x_0)}=I_1+\cO(h^\frac13),
$$
with,
$$
I_1:= \left(2\int_{a(E)}^{c(E)} (E-V_1(x))^{-1/2} dx\right)^{\frac12} = 2\left( {\mathcal A}'(E) \right)^{\frac{1}{2}}.
$$
\end{proposition}
\begin{proof} In view of \eqref{u1L-} and Proposition \ref{wkb-}, our choice of normalization \eq{norm} for $u={}^t(u_1,u_2)$ is made in such a way that, locally on $(-\infty, x_0]$ (where both $f_{1,L}^-$ and $f_{1,L}^+$ contribute to $u$), we have,
$$
|u_1|=\frac{\sqrt \pi}{h^{\frac16}} |u_{1,L}^-| +\ord(\sqrt{h})\quad ; \quad u_2=\ord(\sqrt{h}).
$$ 
Moreover, we know that $u$ is exponentially small for $x<a(E)$ (both in $h$ as $h\to 0_+$, and in $x$ as $x\to -\infty$), and that, if we fix some $x_1<a(E)$, then  $h^{-\frac16}\Vert u_{1,L}^-\Vert_{L^2(x_1,x_0)}=\ord(1)$. Therefore, we have,
$$
\begin{aligned}
\Vert u\Vert_{L^2(-\infty,x_0)}^2= &\frac{\pi}{h^{\frac13}} \Vert u_{1,L}^-\Vert_{L^2(x_1,x_0)}^2 +\ord (\sqrt{h}).
\end{aligned}
$$
Finally, $\Vert u_{1,L}^-\Vert_{L^2(x_1,x_0)}^2$ can be computed by using the results of \cite{FMW1}, and one finds,
$$
\begin{aligned}
&\Vert u_{1,L}^-\Vert_{L^2(x_1,x_0)}^2 \\
& =  \frac{4h^{\frac13}}{\pi}\int_{a(E)+h^{\frac23}}^{c(E)-h^{\frac23}}(E-V_1(x))^{-1/2} \cos^2\left( \frac{{\mathcal A}(E) +\nu_1(x)}{h} -\frac{\pi}4\right) dx+\cO(h).
\end{aligned}
$$
Writing,
$$
2\cos^2\left( \frac{{\mathcal A}(E) +\nu_1(x)}{h} -\frac{\pi}4\right) = 1+\sin \left( \frac{2\left({\mathcal A}(E) +\nu_1(x)\right)}{h} \right),
$$
and making an integration by parts on the second term, we obtain,
$$
\begin{aligned}
\Vert u_{1,L}^-\Vert_{L^2(x_1,x_0)}^2 =&  \frac{2h^{\frac13}}{\pi}\int_{a(E)+h^{\frac23}}^{c(E)-h^{\frac23}}(E-V_1(x))^{-1/2}  dx +\cO(h^{\frac23})\\
=& \frac{2h^{\frac13}}{\pi}\int_{a(E)}^{c(E)}(E-V_1(x))^{-1/2}  dx +\cO(h^{\frac23}).
\end{aligned}
$$
Hence,
$$
\begin{aligned}
\Vert u\Vert_{L^2(-\infty,x_0)}^2 = 2\int_{a(E)}^{c(E)}(E-V_1(x))^{-1/2}  dx +\cO(h^{\frac13}),
\end{aligned}
$$
and the result follows.
\end{proof}

\end{document}